\documentclass[11pt]{article}

\usepackage{amsmath,amsthm,amsfonts,amscd,eucal,latexsym,amssymb,mathrsfs}
\usepackage{amsxtra, calc, bbm}
\usepackage{epsfig}  
\usepackage[latin1]{inputenc}
\usepackage{verbatim} 
\usepackage{graphicx, float}
\usepackage{hyperref}
\usepackage{color}
\usepackage{xcolor}
\usepackage[normalem]{ulem}
\def\WF{{\text{WF}}}

\def\cA{{\ca A}}
\def\cB{{\ca B}}
\def\cC{{\ca C}}

\def\cE{{\ca E}}
\def\cF{{\ca F}}
\def\cG{{\ca G}}

\def\cS{{\ca S}}

\def\bC{{\mathbb C}}           

\def\bN{{\mathbb N}}

\def\bR{{\mathbb R}}

\def\beq{\begin{eqnarray}}
\def\eeq{\end{eqnarray}}

\newcommand{\ca}[1]{{\cal #1}}         

\def\la{\lambda}

\def\om{\omega}

\def\WF{{\text{WF}}}
\def\supp{\textrm{supp}\,}

\newtheorem{theorem}{Theorem}[section]
\newtheorem{proposition}[theorem]{Proposition}

\newtheorem{corollary}[theorem]{Corollary}

\numberwithin{equation}{section}

\def\
vsp{\vspace{0.2cm}}
\def\vspp{\vspace{0.1cm}}

\def\sse #1 {\vsp\ifhmode{\par}\fi\refstepcounter{subsection}
  \noindent {\bf\thesubsection}. {\em #1}.\quad
  \addcontentsline{toc}{subsection}{\protect\numberline{\thesubsection} #1}%
  }

\def\ssb #1 {\vsp\ifhmode{\par}\fi\refstepcounter{subsection}
  \noindent {\bf\thesubsection.} {\bf #1.}\quad
  \addcontentsline{toc}{subsection}{\protect\numberline{\thesubsection} #1}%
  }

\def\ssa #1 {\ifhmode{\par}\fi\refstepcounter{subsection}
  \noindent {\bf\thesubsection.} {\bf #1.}\quad
  \addcontentsline{toc}{subsection}{\protect\numberline{\thesubsection} #1}%
  }

\def\remark #1 {\vsp\vspp\ifhmode{\par}\fi\noindent\noindent {\bf Remark.} {#1}\vsp\vspp\par}
\def\remarks #1 {\vsp\vspp\ifhmode{\par}\fi\noindent\noindent {\bf Remarks.} {#1}\vsp\vspp\par}

\begin{document}



\par
\bigskip
\LARGE
\noindent
{\bf 
Perturbative Algebraic 
Quantum Field Theory on \\Quantum Spacetime: Adiabatic and Ultraviolet \\Convergence}
\bigskip
\par
\rm
\normalsize


\large
\noindent {\bf Sergio Doplicher$^{1}$}, {\bf Gerardo Morsella$^{2}$}
and {\bf Nicola Pinamonti$^{3}$}\\
\par
\small
\noindent $^1$
Dipartimento di Matematica, Universit\`a di Roma ``La Sapienza'', Piazzale Aldo Moro, 5,
I-00185 Roma, Italy, email dopliche@mat.uniroma1.it.\smallskip

\noindent $^2$
Dipartimento di Matematica, Universit\`a di Roma ``Tor Vergata'', Via della Ricerca Scientifica,
I-00133 Roma, Italy, email morsella@mat.uniroma2.it.\smallskip

\noindent $^3$
Dipartimento di Matematica, Universit\`a di Genova, Via Dodecaneso, 35, 
I-16146 Genova, Italy and 
INFN - Sez.\ di Genova, Via Dodecaneso, 33 I-16146 Genova, Italy,
email pinamont@dima.unige.it.\smallskip

 \normalsize
\par
\medskip

\rm\normalsize
\noindent {\small 
March 6, 2020}

\rm\normalsize


\par
\bigskip

\noindent
\small
{\bf Abstract.}
The  quantum structure of Spacetime at the Planck scale suggests the use, in defining interactions between fields, of the Quantum Wick product. The resulting theory is ultraviolet finite, but subject to an adiabatic cutoff in time which seems difficult to remove.
We solve this problem here by another strategy: the fields at a point in the interaction Lagrangian are replaced by the fields at a  {\textit{quantum point}}, described by an  optimally localized state on QST; the resulting Lagrangian density agrees with the previous one after spacetime integration, but gives rise to a different interaction hamiltonian. But now  the methods of perturbative Algebraic Quantum Field Theory can be applied, and produce an ultraviolet finite perturbation expansion of the interacting observables.
If the obtained theory is tested in an equilibrium state at finite temperature 
the adiabatic cutoff in time becomes immaterial, namely it has no effect on the correlation function at any order in perturbation theory. Moreover, the interacting vacuum state can be obtained in the vanishing temperature limit. It is nevertheless important to stress that  the use of  states which are optimally localized for a given observer brakes Lorentz invariance at the very beginning.\\

\normalsize

\section{Introduction}

Quantum Mechanics and Classical General Relativity are the main achievements of the last century in the theoretical description of physical systems. Quantum Mechanics is used to describe subatomic physics while General Relativity is necessary to accurately describe phenomena at the astrophysical level. 
Unfortunately, a universally accepted theory which combines both is still lacking. 
In spite of this, the principles of Quantum Mechanics and of Classical General Relativity meet at least at one crucial point: their concurrence prevents arbitrarily accurate localization of events in Spacetime. 
For the Heisenberg uncertainty principle would imply that with an observation large amount of energy could be packed in small regions;
the corresponding classical backreaction on the curvature could then create trapped regions or even black holes. This limitation on possible localization
manifests itself in Spacetime Uncertainty Relations ~\cite{DFR95}, 
\[
\Delta q_0 \cdot \sum \limits_{j = 1}^3 \Delta q_j \gtrsim \lambda^2 ; \sum   
\limits_{1 \leq j < k \leq 3 } \Delta q_j  \Delta q_k \gtrsim \lambda^2\,.
\]
where $\lambda$ is a constant of the order of the Planck length. The Spacetime Uncertainty Relations can be exactly implemented by Poincar\'e covariant commutation relations between the coordinates 
\[
[q_\mu  ,q_\nu  ]   =   i \lambda^2   Q_{\mu \nu}
\]
if certain  {\textit{Quantum Conditions}} are imposed
on the $Q_{\mu \nu}$. The simplest solution is
\begin{align*}
[q^\mu , Q^{\nu \lambda}]  =  0,\qquad
Q_{\mu \nu} Q^{\mu \nu}  =  0,\qquad
\left(\frac1 2\left[q_0  ,\dots,q_3 \right] \right)^2 = I,
\end{align*}
where, under Lorentz transformations,  $Q_{\mu \nu} Q^{\mu \nu}$ is a scalar, and
\begin{align*}
\left[q_0  ,\dots,q_3 \right]  
\equiv  \varepsilon^{\mu \nu \lambda \rho} q_\mu q_\nu q_\lambda 
q_\rho = - (1/2) Q_{\mu \nu}  (*Q)^{\mu \nu}
\end{align*}
is a pseudoscalar and this is the reason why the Quantum Conditions involve its square.

A model of Quantum Spacetime (QST) can thus be obtained studying the representations of these relations. In particular, their \textit{regular representations} fulfill the \textit{Weyl 
relations}
\[
e^{ih_\mu q^\mu}e^{ik_\nu q^\nu}=e^{-\frac i2\lambda^2 h_\mu Q^{\mu\nu} k_\nu}e^{i(h+k)_\mu q^\mu},\quad h,k\in\mathbb R^4,
\]
which are governed by a $C^*-$algebra $\mathcal E$ which is thus the desired model of QST.  $\mathcal E$ turns out to be  the tensor product of
 the continuous functions vanishing at infinity on $\Sigma$ and the algebra of compact operators.
Here $\Sigma $  - the joint spectrum of the $Q_{\mu \nu}$ in a fully covariant representation - is the manifold of the real valued antisymmetric 2-tensors fulfilling the same relations as the $Q_{\mu \nu}$  do: a homogeneous space of the proper orthochronous Lorentz group, identified with the coset space of $SL(2,\bC)$ mod the subgroup of diagonal matrices.  Each of those tensors can be taken to its rest frame, where the electric and magnetic parts
 $\bf{e}$,   $ \bf{m}$   are parallel unit vectors, by a boost, and go back with the inverse boost, specified by a third vector, orthogonal to those unit vectors; thus  $\Sigma$ can be viewed as the tangent bundle to two copies of the unit sphere in 3 space - its {\textit{base}}  $\Sigma_1$. Irreducible representations select a point $\sigma \in \Sigma$. If  $\sigma \in \Sigma _1$, in that representation the $q_{\mu}$ are, up to a space rotation, the Schr\"odinger operators in two degrees of freedom.

In this model the mathematical generalization of points are pure states over $\mathcal{E}$. Among all possible states over $\mathcal{E}$, we are interested in finding those which describe the best localization. These are called {\textit{Optimally localized states}} and are those which minimize the quantity
\[
 \Sigma_\mu(\Delta_\omega q_\mu)^2.
\]
Its minimum value is $2$ and it is reached by states concentrated on $\Sigma_1$ which at each point coincide with the {\textit{ground state of the harmonic oscillator}}.
More precisely, these states are given by an {\textit{optimal localization map}} composed with a probability measure on  $\Sigma_1$. 
In particular, 
 if $f$ varies in $L^1(\bR^4)$ and $g$ in $\mathcal C_0 (\Sigma)$, the products of the evaluation on the $q_\mu$ of the Fourier pre-transform of $f$  and of $g$ on the $Q_{\mu \nu}$ 
\[
\int_{\bR^4} dk\, f(k)  e^{i k\cdot q} g(Q)
\]
have a dense linear span in $\mathcal E$; states which are {\textit{optimally localized}} at the point $x$ are given by
\[
\om_x(e^{i k_\mu q^\mu}g(Q)) = e^{i k \cdot x}e^{-\frac{\la^2}{2} \langle k\rangle^2} \int_{\Sigma_1} g(\sigma) d\mu (\sigma), \qquad x \in \bR^4,
\]
where $k \cdot x = k_\mu x^\mu$ is the usual Minkowski scalar product, $\langle k\rangle^2 = \sum_{\mu=0}^3 k_\mu^2$ is the Euclidean norm square, and $\mu$ is a regular probability measure carried by $\Sigma_1$. They can be viewed as the composition of $\mu$ with the restriction of $\mathcal E$  to  $\Sigma_1$ followed by the {\textit{optimal localization map}}, a conditional expectation of that restriction onto  $\mathcal C (\Sigma_1)$.

Several reasonable properties follow, in applications to Quantum Geometry in particular: the sum of the squares of the components of the forms describing distance, area, three volume and four - volume in QST are all bounded below by multiples of order one of the appropriate power of the Planck length \cite{BDFP10}.

QFT can also be studied on QST. In the case of a free scalar field $\varphi$, we can define it on QST in analogy to the von Neumann calculus we recalled earlier on:
\[
\varphi(q) := \int_{\bR^4} dk\, \hat \varphi(k) \otimes e^{i k\cdot q},
\]
where $\varphi (q)$ is affiliated to $ \mathfrak F \otimes \mathcal E$, $\mathfrak F$ being the algebra of fields. More generally, and more precisely, the above formula has to be interpreted as an affine map from (a suitable $*-$weakly dense subset of the) states $\omega$ on $\cE$ into the $*-$algebra $\mathfrak F$ of polynomials in the field operators on (commutative) Minkowski space time, given by
\[
\om \mapsto (\mathrm{id}\otimes \om)(\varphi(q)) = \varphi(f_\om) = \int_{\bR^4} dx\,\varphi(x) f_\om(x),
\] 
where the right hand side is the operator on the bosonic Fock space obtained by smearing the ordinary free scalar field with the test function $f_\om(x) := \int_{\bR^4} dk\,\om(e^{ik\cdot q}) e^{-ik\cdot x}$, assuming that $k \mapsto \om(e^{i k\cdot q})$ belongs to the Schwarz space $\cS(\bR^4)$.

To introduce interactions, a possibility is to observe that the Euclidean distance between two independent events on QST cannot be smaller than the Planck length. This motivates the use, in defining interactions between fields on QST, of the Quantum Wick product, where the evaluation of field products at coinciding points is replaced by the use of a quantum diagonal map. 
This map evaluates on the difference variables the optimal localization map.  The resulting theory is ultraviolet finite \cite{BDFP03}, but needs an adiabatic cutoff in time, which seems difficult to remove. On the other hand, the resulting theory has interesting applications to Cosmology~\cite{DMP13, DMP1, DMP2} and generalizations to curved spacetimes are being also studied~\cite{DMP13, TV14, MT}. For other possibilities to define interactions, see~\cite{BDFP05, Z} and the review~\cite{BDMP15}.

The control of the adiabatic limit is typically highly non trivial also in other approaches to QFT on noncommutative spacetime. In particular, in the much studied Euclidean framework based on the modified Feynman rules~\cite{Filk} this difficulty is generally traced back to the phenomenon of the so-called \emph{ultraviolet-infrared mixing}, first pointed out in~\cite{MRS99}. Due to this feature, scalar QFT on Euclidean noncommutative spacetime is typically nonrenormalizable, with the exception of the Grosse-Wulkenhaar model~\cite{GW05} in which however an external harmonic potential provides an effective adiabatic cutoff. Moreover, there is no analogue of the Osterwalder-Schrader theorem allowing a smooth transition between the Euclidean and the Minkowskian regimes, and a Minkowskian version of the Grosse-Wulkenhaar model is affected by strange divergencies~\cite{Z11}. On the other hand, ultraviolet-infrared mixing effects, albeit different from the ones in the Euclidean setting, have also been pointed out in QFT on Minkowskian noncommutative spacetime defined through the Hamiltonian approach~\cite{B10} and through the Yang-Feldman equation~\cite{Z12}. Further discussion of the adiabatic limit of QFT on QST, in the Yang-Feldman approach, can be found in~\cite{DZ, Z}.

\medskip
In order to define QFT on QST without an adiabatic cutoff, it seems natural to try to employ the methods of perturbative Algebraic Quantum Field Theory (pAQFT) \cite{BDF09, FR15a, FR15b, HW01, HW02} (see also the book~\cite{R17} for an extensive review), in which the adiabatic limit on classical spacetime can be controlled for a large class of models \cite{FredenhagenLindner}. However, in doing this one has to face an immediate problem: the nonlocality present in the effective Lagrangian would imply the nonunitarity of the corresponding $S$ matrix defined through the methods of pAQFT.

We solve this problem here by the following strategy: as a first step, the fields at a point in the interaction Lagrangian are replaced by the fields at a quantum point, described by an optimally localized state on QST. More precisely, the evaluation of the field $\varphi$ at a {\textit{quantum point}} of QST, which replaces the field operator $\varphi (x)$ at a point $x$ of the classical Minkowski space, is given by
\begin{equation}\label{eq:quantumpoint}
 (\mathrm{id}\otimes \om_x)(\varphi(q)),
\end{equation}
where $\iota$ is the identity map on $\mathfrak F$, and $\omega_x$ the state of $\mathcal E$ optimally localized at $x$.
 
As discussed in Section \ref{sec:effective}, the resulting Lagrangian density agrees with the one in \cite{BDFP03} after integration over the full spacetime, namely if the adiabatic limit is considered. 
However, it turns out that it is still nonlocal, and then it does not yet yield a unitary $S$ matrix. Here comes the second step of our strategy: the nonlocality of this Lagrangian can be removed in an equivalent nonlocal representation of the field algebra at the price of deforming the products in this algebra. 
In particular, the nonlocalities induced by the smearing of fields around the points $x$ are now attached to the various propagators of the theory. In this way the ultraviolet singularities that the latter present at coinciding points are resolved. Hence an ultraviolet finite perturbation theory is obtained.
The Feynman rules obtained in this way are slightly different from those described by Piacitelli in \cite{piacitelli}, c.f.\ also \cite{Bahns}. Proceeding in this way, the $S$ matrix becomes unitary. 
Furthermore, the obtained theory can be treated with the methods of pAQFT
 which we briefly review in Section \ref{sec:review-paqft}. In particular,   
adapting the analysis of \cite{FredenhagenLindner}, in Section \ref{sec:perturbative-construction} we study, order by order in perturbation theory, the spatial adiabatic limit of interacting KMS states at finite temperature. Moreover, we show that this limit can be taken simultaneously with the zero temperature limit, thereby obtaining the ground state of the interacting theory. 

The latter is the main result of the present work, on which more details are in order. The first crucial observation here is that, although Lorentz covariance is broken by the interaction, there is a remnant of causality in the theory, namely the independence of the interacting observables from the behavior of the infrared cutoff in the future of their localization region. Therefore we can assume that, with a fixed spatial cutoff, the temporal cutoff can be taken equal to one in the future of a given time slice. In order to proceed to the limit in which the temporal cutoff is one everywhere, it is then sufficient to translate this time slice to infinity into the past. To this end, the second main step is to show that the vacuum expectation values of the interacting observables in this limit can also be obtained as the zero temperature limit of the corresponding expectation values in an interacting KMS state. The final step is then to show that the latter limit and the limit in which the spatial adiabatic cutoff is removed can be taken together. The limit state is therefore invariant under spacetime translations, and can then be interpreted as the vacuum state of the interacting theory. This shows in particular that
in the obtained theory there is no ultraviolet-infrared mixing which produces divergencies.

We have nevertheless to stress that the use of  states which are optimally localized for a given observer brakes again Lorentz invariance at the very beginning.

\medskip
The paper is organized as follows. In the next section we briefly review the main ideas of pAQFT. Section~\ref{sec:effective} contains the derivation of the effective nonlocal Lagrangian and its comparison with previously used expressions. In Section~\ref{sec:perturbative-construction} we describe the perturbation theory and we construct perturbatively the $S$ matrix and the interacting field algebra. In Section~\ref{sec:adiabatic} we construct the ground state of the interacting theory by showing that the adiabatic limit is order by order finite in the sense discussed above. We collect in the appendices some technical results. In particular, in Appendix~\ref{se:prop} we list some formulas and elementary properties of the modified propagators, in Appendix~\ref{sec:bounds} we discuss their decay properties in spatial and temporal directions, and in Appendix~\ref{sec:KMS} we prove the return to equilibrium property for the spatially cutoff KMS state of the interacting theory.

\section{Brief review of pAQFT for scalar fields}\label{sec:review-paqft}

We recall briefly the framework of pAQFT applied to discuss a self interacting scalar field theory, we refer to \cite{BDF09, FR15a, FR15b, HW01, HW02} for an extended review.
Although the construction we are going to describe can be performed in any globally hyperbolic spacetime, we shall here only consider a four dimensional Minkowski spacetime $M=\bR^4$  equipped with a metric $\eta$ whose signature is $(-,+,+,+)$ as the background where a scalar field propagates.
The Lagrangian density for the theory we are considering is described by 
\[
L =L_0+L_I = \frac{1}{2} \partial\phi\partial\phi +\frac{1}{2} m^2\phi + L_I(\phi)
\]
where $m > 0$ is the mass of the field and $L_I$ represents the interaction Lagrangian which is usually a local polynomial built from $\phi$ at a point $x$. We shall construct the interacting theory perturbing the free theory $L_0$ with the interaction Lagrangian $L_I$.

The observables of the theory are seen as functionals over smooth {\it field configurations} $\phi\in C^\infty(M;\mathbb{R})\cap \mathcal{S}'(M)=:\mathcal{C}$.
For later purpose, contrary to the standard framework, we require that the field configuration $\phi$ is a Schwartz distribution. We require furthermore, that the observables are smooth with respect to the functional derivatives and that they have compact support, namely their $n-$th order functional derivative exists and is a distribution of compact support. Moreover, they must be sufficiently regular to make their product well defined. In particular, we say that a smooth functional $F$ is {\it microcausal} ($\mu$c), if the wave front sets of its functional derivatives are such that
\[
\WF(F^{(n)}(\phi)) \cap \left({\overline{V}^+}^n \cup  {\overline{V}^-}^n \right) = \emptyset, \qquad n \in \bN,
\]
where $\overline{V}^\pm$ are the closed forward backward light cones in $T^*M$.  

Summarizing, the {\it field observables} are defined as the set
\[
\mathcal{F} := \left\{ F:\mathcal{C}\to\mathbb{C} \left| F^{(n)}(\phi)\in \mathcal{E}'(M^n) \, \forall \phi \in \cC, \forall n\in \bN,\, F \text{ $\mu$c }   \right. \right\}.
\] 
Furthermore, we denote by $\mathcal{F}_{\text{reg}}$ the set of {\it regular functionals}, namely
\[
\mathcal{F}_{\text{reg}} := \left\{ F\in \mathcal{F} \left| F^{(n)}(\phi) \in C^\infty_0(M^n) \, \forall \phi \in \cC,\forall n \in \bN\  \right. \right\}.
\]
For later purposes, we also introduce the slightly larger space of \textit{rapidly decreasing functionals}
\[
\cF_{\text{rd}} := \left\{ F : \cC \to \bC\left| F^{(n)}(\phi) \in \cS(M^n)\, \forall \phi \in \cC,\forall n \in \bN\right.\right\}.
\]
We say that $F\in \mathcal{F}$ is {\it local} if $F^{(1)}(\phi)$ is a smooth function and if for every $n>1$ the $n-$th functional derivative of $F$ is supported on the diagonal of $M^n$. The subset of $\mathcal{F}$ formed by {\it local functionals} is denoted by $\mathcal{F}_{\text{loc}}$
Finally a functional $F : \cC \to \bC$ is said to be {\it polynomial} if it exists an $N$ such that $F^{(n)}=0$ for all $n>N$. 
The subsets of $\mathcal{F}$, $\cF_\text{rd}$ and of $\mathcal{F}_{\text{reg}}$ formed by polynomial functionals are denoted by $\mathcal{F}^p$, $\cF^p_{\text{rd}}$ and $\mathcal{F}^p_{\text{reg}}$ respectively.  
The interaction Lagrangian $L_I$ smeared with a compactly supported smooth function $g\in C^\infty_0(M)$ is an example of a local functional
\begin{equation}\label{eq:int-L}
V = \int_M g(x)   L_I(\phi)(x)  dx.
\end{equation}
A local linear field smeared with $f\in C^\infty_0(M)$ and its Wick powers are then defined as 
\begin{equation}\label{eq:linear-wick}
\Phi(f):= \int_M f(x) \phi(x) dx, \qquad
\Phi^n(f):= \int_M f(x) \phi^n(x) dx, \qquad n\in \mathbb{N}.
\end{equation}

\subsection{Algebra of free field observables}
The construction of the free theory starts from the observation that the free equation of motion
\begin{equation}\label{eq:free-equation}
\Box\varphi - m^2\varphi=0
\end{equation}
admits unique retarded and advanced fundamental solutions $\Delta_{R/A}$. 
These distributions are completely characterized by their interplay with the equations of motion and by their support property \cite{BGP}, namely  
\[
\Delta_{R/A}(\Box f - m^2 f) =  f, \qquad    \text{supp}(\Delta_{R/A}(f)) \subset J^{+/-}(\text{supp}\,f), \qquad f\in C^\infty_0(M),
\]
where $\text{supp}\,f\subset M$ is the support of $f$, and $J^{+/-}(O)$ is respectively the causal future or the causal past of the set $O$ in $M$.  The commutator function, also known as Pauli-Jordan propagator or causal propagator, is the retarded minus advanced fundamental solution
\[
\Delta = \Delta_R-\Delta_A
\] 
and in the canonical quantization of the free theory it corresponds to the commutator of two linear fields.
Hence, in the algebraic approach, the set of functionals described above are equipped with a product whose antisymmetric part is completely determined by $\Delta$.

\medskip
More precisely, both $\mathcal{F}^p$ and $\mathcal{F}^p_{\text{reg}}$ form a $*-$algebra when equipped with the following involution and $\star-$product
\[
A^*(\phi) =  \overline{A(\phi)}, \qquad A\in \mathcal{F}^p
\]
and
\[
A\star_H B =  m \circ e^{ \Gamma_{H}}   (A\otimes B)
\]
where 
$m$ is the pointwise multiplication
\[
m(A\otimes B)(\phi) = A(\phi)B(\phi), \qquad A,B\in \mathcal{F}^p
\]
and $\Gamma_H$ is such that
\[
\Gamma_H = \int_{M^2} dxdy H(x,y)\frac{\delta}{\delta \phi(x)}\otimes \frac{\delta}{\delta \phi(y)},
\]
where $H \in \cS'(M^2)$ is any Hadamard function, namely a weak solution of the free equation of motion up to smooth terms whose antisymmetric part is proportional to the causal propagator $\Delta$, which is the retarded minus advanced fundamental solutions
\[
H(x,y)- H(y,x) = i\Delta(x-y) = i \Delta_R(x-y)- i \Delta_A(x-y). 
\]
Furthermore, $H$ must satisfy the microlocal spectrum condition  which means that 
\[
\WF(H) = \{(x_1,x_2;p_1,p_2)\in T^*M^2\setminus\{0\} \,|\, (x_1,p_1)\sim (x_2,-p_2), (x_1,p_1)\in  \overline{V}^+  \}.
\] 
where $(x_1,p_1)\sim(x_2,-p_2)$ holds if $x_1$ and $x_2$ are joined by a null geodesics $\gamma$, $\eta^{-1}p_1$ is tangent to $\gamma$ at $x_1$ and $-p_2$ is the parallel transport of $p_1$ along $\gamma$. The last condition requires $p_1$ to be future directed.
This requirement can be weakened if we are considering regular functionals only. In that case we may use the product $\star_{\frac{i}{2}\Delta}$ defined by $ H = \frac i 2 \Delta$. We also observe that the product $\star_{\frac i 2 \Delta}$  is well defined on the space $\cF^p_{\text{rd}}$ of rapidly decreasing functionals too, which then becomes a *-algebra with  the above defined involution.

Notice that $(\mathcal{F}^p,\star_H)$ is a representation of the abstract algebra of {\it normal ordered observables}  $\mathcal{B}$, where the normal ordering is the one defined with respect to the Hadamard state specified by the two point function given by $H$. Indeed, 
the $*-$isomorphism which realizes the representation of $\mathcal{B}$ to $(\mathcal{F}^p,\star_H)$
maps $:\Phi^n(f):_H$ to $\Phi^n(f)$, and 
the combinatorics of Wick Theorem is summarized by the product $\star_H$, e.g., 
\[\begin{split}
:\Phi^2(f):_H:\Phi^2(f):_H   &=   :\Phi^2(f) \star_H \Phi^2(g):_H  \\
&= 
:\Phi^2(f)\Phi^2(g):_H + 4\int_{M^2} dx dy\, :\phi(x)\phi(y):_H f(x) g(y) H(x-y) \\
&\quad +2 \int_{M^2} dx dy\, f(x) g(y) H(x-y)^2. 
\end{split}\]
We remark however that in general the elements of $\cB$ can not be realized as functionals of the field configurations. Because of this, we prefer to consider only the algebra $(\mathcal{F}^p,\star_H)$.

At the same time $(\mathcal{F}^p_{\text{\text{reg}}},\star_{i\Delta/2})$ is isomorphic to the {\it Borchers-Uhlmann algebra} of free fields. In particular, the canonical commutations relations among linear fields are implemented by the product, actually
\[
[\Phi(f_1),\Phi(f_2)]_{\star_{H}} = \Phi(f_1)\star_{H}\Phi(f_2) - \Phi(f_2)\star_{H}\Phi(f_1) = i \Delta(f_1,f_2), \qquad f_i\in C^\infty_0(M).
\]

At this point it is important to stress that different choices of $H$ in $(\mathcal{F}^p,\star_H)$ produce isomorphic algebras, we indicate this $*-$isomorphism by $\alpha_{H,H'}$. 

We furthermore notice that the restriction to polynomial functionals can be relaxed if we admit as observables formal power series in $\hbar$, or if we require some bound on the higher functional derivatives of $F$ in the topology of distribution  so to make $A\star_H B$ convergent.

\subsection{Interacting fields}\label{se:paqft}

In perturbation theory, the algebra of interacting fields with respect to a local interaction Lagrangian $V$, like the one given in \eqref{eq:int-L}, is represented as a subalgebra of the algebra of formal power series whose coefficients are in $\mathcal{F}$. 
This representation is realized by the Bogoliubov map $R_V$. The Bogoliubov map is given in terms of the time ordered exponential ($S$ matrix) of $V$, defined by
\begin{equation}\label{eq:S}
S(V)=\sum_n \frac{(-i)^n}{n!} \underbrace{T( T^{-1}(V),\dots ,T^{-1}(V) )}_{n}.
\end{equation}
As shown by Fredenhagen and Rejzner (see, e.g.,~\cite{FR15a}) the time ordering map $T$ applied to regular functionals descends from an associative bilinear product
\[
A\cdot_T B = T(T^{-1}(A),T^{-1}(B))
\] 
which has to be compatible with the $\star$ product introduced above. Hence, it holds that
\[
A\cdot_T B =  M \circ e^{ \Gamma_{H_F}}   (A\otimes B), \qquad A, B \in \cF_{\text{reg}},
\] 
where the Feynman propagator associated to the Hadamard function $H$ is used, actually 
\[
H_F = H + i\Delta_A.\footnote{Sometimes, in the physics literature, the Feynmann propagator, with respect to the Minkowski vacuum state, is defined as $-iH_F$.}
\]
On regular functional if the causal propagator $\frac{i}{2}\Delta$ is used to represent the $\star$ product, the associated time ordered product is generated by the Dirac propagator  
\[
i\Delta_D = i(\Delta_A + \Delta_R).
\] 
In order to give a meaning to~\eqref{eq:S} for a local $V$, the time ordering map $T$ introduced above needs to be extended to local functionals. 
This extension can be performed applying a recursive procedure on the number of factors in $T$, in a similar way as discussed by Epstein and Glaser \cite{EG}. The key property that makes this possible is the \emph{causal factorization property} which requires that 
\[
T(A,B)=T(A)\star T(B)
\]
if $J^{+}(\supp{A}) \cap \supp B =\emptyset$. At every inductive step, the causal factorization property permits to construct $T(V,\dots, V)$ as a distribution over the coupling constants $g\otimes \dots \otimes g$ defined up to the total diagonal $\{(x,\dots, x) \in M^n\,|\, x\in M\}$. The extension to the total diagonal can be done employing some scaling limit techniques similar to the procedure discussed by Steinmann~\cite{Steinmann}, see e.g.\ \cite{BF00, HW02} for a detailed analysis of this procedure.
The extension to the total diagonal is not unique, there is actually the well known renormalization freedom.
This freedom can be characterized by some constants once a set of physical requirements for the time ordered products are assumed, see \cite{HW01, HW02} for a complete list of axioms that a time ordered product should have.

Later we shall adapt the causal factorization property to 
take into account the noncommutative nature of Quantum Spacetime and its replacement will play a crucial role in  the discussion of the adiabatic limits.
\medskip

The algebra of interacting fields is thus generated, with respect to the $\star$ product, by elements of the form
\begin{equation}\label{eq:Bog}
R_V(F)=S(V)^{-1}\star \left(F \cdot_T S(V) \right), \qquad  F\in\mathcal{F}_{\text{loc}},
\end{equation}
where $R_V$ is the \emph{Bogolubov map}. If $V$ is local it holds that $S(V)^{-1} = S(V)^*$, hence, in that case we say that the theory is unitary. This is in general not true if $V$ is nonlocal.

\section{The effective interaction Lagrangian obtained by the quantum Wick product}
\label{sec:effective}

As discussed in \cite{BDFP03}, 
a possibility to treat the interacting quantum field theory on QST is to define the interaction Lagrangian replacing the ordinary Wick product (on commutative spacetime) with the quantum Wick product $: \,\cdot\, :_Q$, which amounts to evaluate the optimal localization map on the difference variables of a product of fields on independent copies of QST. The resulting $S$ matrix is then equivalent to the $S$ matrix of a theory on classical spacetime with an effective interaction Lagrangian $V$ which is usually nonlocal in the field.

More precisely, as shown in \cite{BDFP03} the effective interaction Lagrangian corresponding to a monomial interaction $:\varphi^n(q):_Q$ has the expression
\[
L_I^{\text{eff}}(x):=\frac{(2\pi)^{2(n+1)}}{\lambda^{4(n-1)}n^2} \int dy_1\dots dy_n r(y_1,\dots, y_n,x):\phi(y_1)\dots \phi(y_n):,
\]
where
\[
r(y_1,\dots, y_n,x) :=
e^{-  \frac{\langle x-y_1\rangle^2}{2 \lambda^2}}
\dots
e^{-  \frac{\langle x-y_n\rangle^2}{2 \lambda^2}}
\delta\left(\frac{y_1+\dots +y_n}{n}-x\right),
\]
and where $\langle x\rangle^2= ({x^0})^2+ |{\mathbf{x}}|^2$ denotes the euclidean square and $\delta$ denotes the four dimensional Dirac delta function.
Integrating against a compactly supported smooth function $\tilde{g}\in C^\infty_0(M)$, which plays the role of an infrared regulator, and defining the time ordering with respect to the time coordinate of the ``average vertex'' $x$, one obtains an $S$ matrix which turns out to be unitary and free of ultraviolet divergencies. On the other hand, in view of studying this theory by the methods of pAQFT, one should consider an effective interaction potential of the form 
\begin{equation}\label{eq:effpot}
W_{\text{eff}}(\phi) := \int dx\;\tilde{g}(x)  L^{\text{eff}}_I(x)
\end{equation}
which is an element of $\mathcal{F}^{p}_{\text{rd}}$ but is nonlocal, which implies that the associated $S$ matrix~\eqref{eq:S}, defined through the time ordering of Sec.~\ref{se:paqft}, is nonunitary. 

In order to circumvent this problem, we observe that 
performing the integration over $x$ we obtain
\begin{align*}
W_{\text{eff}}(\phi) 
 &= \frac{(2\pi)^{2(n+1)}}{\lambda^{4(n-1)}n^2}\int dy_1 \dots  dy_n 
e^{-  \frac{\langle y_1-\overline{y}\rangle^2}{2 \lambda^2}}
\dots
e^{-  \frac{\langle y_n-\overline{y}\rangle^2}{2 \lambda^2}} 
\tilde{g}(\overline{y}):\phi(y_1)\dots \phi(y_n): 
\end{align*}
where $\overline{y} := \frac{1}{n}\sum_j y_j$.
Notice furthermore that
\begin{equation}\label{eq:id}
\sum_j \langle y_j-\overline{y}\rangle^2 +n\langle x-\overline{y}\rangle^2 = \sum_j \langle y_j-x\rangle^2,
\end{equation}
hence we can simplify some of the Gaussian functions present in the previous integral. 
Actually, since for every $\overline{y}$ 
\begin{equation}\label{eq:gaus}
\frac{n^2}{(\sqrt{2\pi}\lambda)^4}\int dx \; e^{-  \frac{n \langle x-\overline{y}\rangle^2}{2 \lambda^2}} = 
1,
\end{equation}
we have that
\begin{equation*}
W_{\text{eff}} (\phi)
 =\left( \frac{\sqrt{2\pi}}{\lambda}\right)^{4n} \int dy_1 \dots  dy_n dx\;
e^{-  \frac{\langle y_1-x\rangle^2}{2 \lambda^2}}
\dots
e^{-  \frac{\langle y_n-x\rangle^2}{2 \lambda^2}} 
\tilde{g}(\overline{y}):\phi(y_1)\dots \phi(y_n): \,,
\end{equation*}
where we have multiplied $g(\overline{y})$ with the Gaussian displayed on the left hand side of the equation \eqref{eq:gaus} and we have used the identity \eqref{eq:id}.
We thus conclude that in the adiabatic limit ($\tilde{g}\to 1$) the effective potential $W_{\text{eff}}$ is formally equivalent to the limit ${g}\to1$ of the following alternative effective potential
\begin{equation*}
{V}_{\text{eff}} (\phi)= 
\left( \frac{\sqrt{2\pi}}{\lambda}\right)^{4n}
\int  dx  \; g(x) \prod_j \left(\int  dy_j \; e^{-\frac{ \langle y_j-x\rangle^2}{2 \lambda^2 }}
\phi(y_j)\right)
\end{equation*}
which is another element of $\mathcal{F}^p_{\text{rd}}$.
We are eventually interested in taking the adiabatic limit, furthermore as we shall see later, the effective potential
${V}_{\text{eff}}$ give rise to a simpler interacting theory. Hence we shall discuss perturbation theory with respect to this potential.
This effective potential is still nonlocal, however, as we shall see below, it is possible to modify the definition of time ordered products to obtain an $S$ matrix which is unitary.

\section{Perturbative analysis of the effective interacting theory}\label{sec:perturbative-construction}
As discussed in the previous section, we will replace, in the construction of the interacting $\phi^n$ theory on Quantum Spacetime,
the effective interaction Lagrangian  which correspond to the interaction Lagrangian $:\varphi^n(q):_Q$ on the QST with
\begin{equation}\label{eq:eff-adiab}
{V}_{\text{eff}}(\phi) = 
(2\pi)^{4n}
\int  dx  \; g(x) \prod_j
\left(\int_{M} dy_j \,G_\lambda(x-y_j)\phi(y_j)\right)
\end{equation}
where we have inserted an infrared regulator 
$g \in C^\infty_0(M)$ that we shall eventually remove considering the limit $g\to 1$,
and 
where
\[
G_\la(x) := \frac {e^{-\frac{\langle x\rangle^2}{2\la^2}}}{(\sqrt{2\pi}\la)^4}.
\]
We recall that a direct application of the perturbation theory discussed in Section \ref{se:paqft} for nonlocal interaction Lagrangian gives rise to an $S$ matrix which is nonunitary. 
We shall actually see that this problem can be cured by deforming the $\star_{\frac i 2 \Delta}$ product of $\cF^p_{\text{rd}}$. 
We notice that the effective nonlocal interaction potential $V_\text{eff}$ can be seen as the action of
a local functional 
composed with the convolution of the field configuration $\phi$ with $G_\lambda$.  To formalize this idea, 
consider the map $\iota:\mathcal{C}\to \mathcal{C}$ whose action on a field configuration $\phi\in\mathcal{C}$ is 
\[
(\iota \phi) (x) := \int_M  dy   \;  G_\lambda(x-y) \phi(y).
\]
Note that  $\iota \phi(x)$ can be understood as a field at a quantum point of QST, described in 
\eqref{eq:quantumpoint}, namely 
\begin{equation}\label{eq:phiomega}
(\iota \varphi) (x)=\int_{M} dy\, G_\la(x-y)\varphi(y) = (\mathrm{id}\otimes\om_x )(\varphi(q)).
\end{equation}
 
The pullback of $\iota$ on field observables is denoted by $r_\lambda$ and it acts on $F \in  \mathcal{F}$ as
\[
r_\lambda F(\phi) := F(\iota \phi).
\]
It is then easy to see that the Fourier transforms of the functional derivatives of $r_\lambda F$ satisfy
\[
\widehat{(r_\lambda F)^{(n)}(\phi)}(p_1,\dots,p_n) = \widehat{F^{(n)}(\iota \phi)}(p_1,\dots,p_n) e^{-\lambda^2\sum_j\langle p_j\rangle^2},
\]
and therefore, by the Payley-Weiner theorem for compactly supported distributions, $r_\lambda$ maps $\cF$ into $\cF_{\text{rd}}$.
We thus observe that
\[
V_{\text{eff}} =  (2\pi)^{4n}r_\lambda(V_g)
\]
where $V_g$ is the local functional
\begin{equation}\label{eq:V}
V_g(\phi) = \int_M g(x) \phi^n(x)   dx,
\end{equation}
Hence, the effective potential is the image under $r_\lambda$ of a local interaction Lagrangian. The construction of time ordered product among local functionals leads to unitary $S$ matrices. While the $S$ matrix constructed out of a nonlocal potential, though regular, is in general not unitary. 
We furthermore observe that similarly to \eqref{eq:phiomega},  
\[
V_{\text{eff}} =  (2\pi)^{4n}
\int_M  dx  \; g(x) \;(\mathrm{id}\otimes\omega_x)(\varphi(q))^n.
\]

\medskip
For this reason we  evaluate the change in the algebraic product of $(\mathcal{F}^p_{\text{reg}},\star_{i\Delta/2})$ under the action of $r_\lambda$, so to see $r_\lambda$ as a $*$-homomorphism. Then we can write the Bogoliubov map and thus its perturbative expansion before applying $r_\lambda$.
We start by observing that
\begin{equation}\label{eq:starlambda}
r_\lambda F_1 \star_{i\Delta/2} r_\lambda F_2 = r_\lambda(F_1\star_{i{\Delta}_\lambda/2} F_2) , \qquad F_i\in \mathcal{F}^p_{\text{reg}},
\end{equation}
where the modified causal propagator ${\Delta}_\lambda$ is such that
\begin{equation}\label{eq:deltalambda}
\Delta_\la(x) := \int_{M^2} dydz\, G_\la(x-y)\Delta(y-z) G_\la(z) \qquad \Leftrightarrow  \qquad
\hat{\Delta}_\lambda(p) =  e^{-\lambda^2 \langle p\rangle^2}\hat{\Delta}(p),
\end{equation}
which entails $\Delta_\lambda \in C^\infty(M) \cap \cS'(M)$. Notice that commutator of fields at different quantum points on QST is such that 
\[
[(\mathrm{id}\otimes\om_x)(\varphi(q)), (\mathrm{id}\otimes\om_y)(\varphi(q))] = i \Delta_\la(x-y).
\]

We furthermore observe that since the product $\star_{\frac i 2 \Delta}$ is well defined on $\cF^p_{\text{rd}}$, Equation~\eqref{eq:starlambda} can be used to extend the product $\star_{\frac i 2 \Delta_\lambda}$ to $\cF^p$. Moreover, the involution commutes with $r_\lambda$, namely $r_\lambda(F^*)=(r_\lambda F)^*$ for every $F\in\mathcal{F}^p$. For these reasons $r_\lambda$ can be seen as a $*-$homomorphism
\[
r_\lambda:(\mathcal{F}^p,\star_{i{\Delta}_\lambda/2})\to (\mathcal{F}^p_{\text{rd}},\star_{i{\Delta}/2}).
\]
We may thus study the perturbation theory in $(\cF^p, \star_{\frac i 2 \Delta_\la})$ before applying $r_\lambda$ so that the interaction Lagrangian is local. We shall eventually obtain the effective algebra which might be smaller than $\cF^p_{\text{rd}}$ applying at the very end $r_\lambda$ (if necessary).
\medskip

At this point we notice that if we want to use the $\star$ product constructed with the vacuum two-point function $\Delta_+$, so to consider the extended algebra of Wick polynomials, a similar transformation holds, namely
\[
r_\lambda:(\mathcal{F}^p,\star_{{\Delta}_{+,\lambda}})\to (\mathcal{F}^p_{\text{rd}},\star_{{\Delta_+}}),
\]
where, similarly to \eqref{eq:deltalambda}
\begin{equation}\label{eq:delta+lambda}
\Delta_{+,\la}(x) := \int_{M^2} dydz\, G_\la(x-y)\Delta_+(y-z) G_\la(z) \qquad \Leftrightarrow  \qquad
\hat{\Delta}_{+,\lambda}(p) =  e^{-\lambda^2 \langle p\rangle^2}\hat{\Delta}_+(p).
\end{equation}
We furthermore remark that as usual
\[
\Delta_{+,\lambda}(x-y)-\Delta_{+,\lambda}(y-x) = i \Delta_\lambda(x-y).
\]
Some basic properties of the functions $\Delta_\la$ and $\Delta_{+,\la}$ are stated in Prop.~\ref{pr:prop} in Appendix~\ref{se:prop}.

Thanks to the fact that $r_\lambda$ is a $*-$homomorphism, we may construct the perturbative series, and hence the time ordering and the Bogoliubov map, directly in $(\mathcal{F}^p,\star_{{\Delta}_{+,\lambda}})$. Notice that, since in $(\mathcal{F}^p,\star_{\Delta_{+,\lambda}})$ the effective potential is described by a local functional, no problem with unitarity should appear if the time ordering with respect to ${\Delta}_{+,\lambda}$ or $i{\Delta_\lambda}/2$ is considered. 
Furthermore, due to the presence of the modified propagators, no ultraviolet singularities arises in the definition of the time ordering product. To prove this in detail we construct the time ordered (Feynman) propagators with respect to ${\Delta}_{+,\lambda}$ which is defined as 
\begin{equation}\label{eq:feynman}
{\Delta}_{F,\lambda}(x-y) := \theta(x^0-y^0){\Delta}_{+,\lambda}(x-y)+\theta(y^0-x^0){\Delta}_{+,\lambda}(y-x). 
\end{equation}
We may thus introduce also the advanced, retarded and Dirac modified propagators as
\begin{align*}
\Delta_{A,\lambda} &:= -i \left(\Delta_{F,\lambda} - \Delta_{+,\lambda}\right) \\
\Delta_{R,\lambda} &:= \Delta_{A,\lambda}+{\Delta}_\lambda \\
\Delta_{D,\lambda} &:= \frac{i}{2}\left(\Delta_{A,\lambda}+{\Delta}_{R,\lambda} \right)
\end{align*}

\begin{proposition}\label{pr:prop1}
The Feynmann propagator on Quantum Spacetime $\Delta_{F,\lambda}$ is a continuous bounded function.
Furthermore, 
\[
(\Box-m^2)\Delta_{F,\lambda}(x) =  -i2\sqrt{2\pi}\lambda\delta(x^0)
G_{2\lambda}(x)e^{-\lambda^2 m^2}.
\]
\end{proposition}
\begin{proof}
From the definition of $\Delta_{F,\lambda}$, we have that
\[\begin{split}
{\Delta}_{F,\lambda}(x) &= \theta(x^0){\Delta}_{+,\lambda}(x)+\theta(-x^0){\Delta}_{+,\lambda}(-x)
= {\Delta}_{+,\lambda}(x)+i\theta(-x^0){\Delta}_{\lambda}(-x)
\\
&= {\Delta}_{+,\lambda}(-x)+i\theta(x^0){\Delta}_{\lambda}(x). 
\end{split}\]
From Proposition \ref{pr:prop} we have that both ${\Delta}_{+,\lambda}$ and ${\Delta}_{\lambda}$ are smooth functions,
furthermore from Eq.~\eqref{eq:delta}, 
\begin{equation}\label{eq:Delta0}
{\Delta}_\lambda(0,\mathbf{x}) = \frac{1}{(2\pi)^3} 
\int _{\bR^3}d\mathbf{p} 
\frac{\sin (\mathbf{x}\cdot \mathbf{p})}{\omega(\mathbf{p})}e^{-\lambda(|\mathbf{p}|^2+\omega(\mathbf{p})^2)}=0
\end{equation}
by parity, hence $\Delta_{F,\lambda}$ is a continuous function.  
The $L^\infty$ norm of $\Delta_{F,\lambda}$ is controlled by the $\|\Delta_{+,\lambda}\|_\infty$ which has been proved to be finite in Proposition \ref{pr:prop}.
To conclude the proof we notice that
\[
(\Box-m^2)\Delta_{F,\lambda}(x) = i(\Box-m^2) \theta(x^0)\Delta_{\lambda}(x)  = 
-i\delta(x^0)\partial_{x^0}\Delta_\lambda(0,\mathbf{x}),
\]
where we used the fact that $\delta(x^0) \Delta_\lambda(x) = 0$ and $(\Box-m^2)\Delta_\la = 0$.
Finally we compute
\begin{align*}
\partial_{x^0}\Delta_\lambda(0,\mathbf{x})  &= \frac{1}{(2\pi)^3} \int_{\mathbb{R}^3} e^{-\lambda^2(\omega(\mathbf{p})^2+|\mathbf{p}|^2)} \cos(\mathbf{p}\cdot \mathbf{x}) d\mathbf{p} =  \\
& =
\frac{1}{(2\pi)^3} \left(\frac{1}{\la}\sqrt{\frac\pi 2}\right)^3 e^{-\frac{|\mathbf{x}|^2}{8\lambda^2}} e^{-\lambda^2 m^2}
=2\sqrt{2\pi}\la G_{2\lambda}(0,\mathbf{x})e^{-\lambda^2 m^2}
\end{align*}
which entails the formula in the statement.
\end{proof}

Notice that causality and Lorentz invariance are broken by the modification introduced above, however
we have that 
\begin{align*}
\supp \Delta_{A,\lambda} &\subset\{(t,\mathbf{x}) \in M\,|\, t\leq 0\},   
\\
\supp \Delta_{R,\lambda} &\subset\{(t,\mathbf{x})\in M\,|\, t\geq 0\}. 
\end{align*}

\begin{proposition}
The distribution $\Delta_{A,\lambda} {\Delta}_{R,\lambda }$ vanishes identically.
\end{proposition}
\begin{proof}
Because of the support property of both  $\Delta_{A,\lambda}$ and $\Delta_{R,\lambda}$, the support of 
$\Delta_{A,\lambda} {\Delta}_{R,\lambda}$ is contained in $\{(0,\mathbf{x})\}$, however, as shown in \eqref{eq:Delta0},
$\Delta_\lambda(0,\mathbf{x})=0$. Proposition \ref{pr:prop} assures that $\Delta_\lambda$ is a smooth function, hence the definition of $\Delta_{A/R,\lambda}$ as $\theta{(\pm x^0)}\Delta_\lambda(x)$ implies the thesis.
\end{proof}

From now on, to simplify the notation we set $\star_\lambda:=\star_{\Delta_{+,\lambda}}$, and we consider the theory 
$\cA(O)$, $O \subset M$ open, whose elements are finite sums of multilocal functionals
\begin{equation}\label{eq:A}
A(\phi) = \int_{M^\ell} dx_1\dots dx_\ell \,f(x_1,\dots,x_\ell) \phi(x_1)^{n_1}\dots \phi(x_\ell)^{n_\ell},
\end{equation}
where $f \in C_0(O^\ell)$, $\ell=1,\dots, N$. Thanks to the fact that $\Delta_{+,\lambda} \in C^\infty(M)$, it is easy to see that $\cA(O)$ is a subalgebra of $(\cF^p,\star_\lambda)$. We also set $\cA_{\text{loc}}(O) := \cA(O) \cap \cF_{\text{loc}}$, whose elements are clearly linear combinations of Wick monomials without derivatives $\Phi^n(f)$ with $f \in C^\infty_0(O)$. We shall also use the notation $\mathcal{A}=\mathcal{A}(M)$.
Notice that $\Phi^n(f) \in \mathcal{A}$ are considered to be normal ordered with respect to the two point function $\Delta_{+,\lambda}$. 

The time ordered product
$T_\lambda:=T_{{\Delta}_{F,\lambda}}$ with respect to $\star_\lambda$ is thus the bilinear map $\cdot_{T_\lambda}: \cA \times \cA \to \cA$ defined as 
\[
F_1\cdot_{T_\lambda} F_2 =  m \circ e^{ \Gamma_{i\Delta_{F,\lambda}}}   (F_1\otimes F_2), \qquad F_i\in\cA
\]
where as seen in Prop.~\ref{prop:Feynman} in Appendix \ref{se:prop}
\begin{equation}\label{eq:modFeyn}
\hat{{\Delta}}_{F,\lambda}(p) =
\frac{-i}{(2\pi)^4}
 \frac{1}{ p^2 +m^2 -i\epsilon}
 e^{-\lambda^2 (2|\mathbf{p}|^2+m^2)}.
\end{equation}
The fact that $\cdot_{T_\lambda}$ is well defined on $\cA$ is an immediate consequence of the continuity of $\Delta_{F,\lambda}$ stated in Prop.~\ref{pr:prop1}, and since $\Delta_{F,\lambda}(x) = \Delta_{F,\lambda}(-x)$, $\cdot_{T_\lambda}$ is commutative. We remark that the analogue of Eq.~\eqref{eq:starlambda} does not hold if one replaces $\star_{i \Delta/2}$ with $\cdot_{T}$ and $\star_{i \Delta_\la/2}$ with $\cdot_{T_\la}$, i.e., the theory we are considering is not equivalent to a theory on classical spacetime with the nonlocal interaction~\eqref{eq:eff-adiab}.
Moreover, contrary to theories on classical spacetime, the products $\phi(x_1)\cdot_{T_\la}\dots\cdot_{T_\la} \phi(x_n)$ are well defined for all $x_1,\dots, x_n \in M$ (also coinciding). In other words, the factor $e^{-\lambda (2|\mathbf{p}|^2+m^2)}$ plays the role of an ultraviolet regulator, so that no ambiguities remains in the definition of the time ordered product.  We also observe explicitly that because of the presence of $\theta$-functions in Eq.~\eqref{eq:feynman} it is clear that in order to extend $\cdot_{T_\lambda}$ to functionals depending on timelike derivatives of $\phi$
a renormalization would be needed. However, as we are only interested in defining the $S$ matrix for an interaction given by a Wick polynomial,  and in applying the Bogoliubov map to functionals in $\cA$, we will refrain from discussing this point further. 

The $S$ matrix is then defined as the formal power series  (in the ``coupling constant'' $g(0)$) with coefficients in $\cA$ given by the $T_\la$-ordered exponential of the interaction $V_g$ in~\eqref{eq:V}, which is an element of $\cA$:
\begin{align*}
S(V_g) &= \sum_{k=0}^{+\infty} \frac{(-i)^k}{k!}V_g\cdot_{T_\la} V_g \cdot_{T_\la} \dots \cdot_{T_\la} V_g  \\
&=  \sum_{k=0}^{+\infty} \frac{(-i)^k}{k!} \int_{M^n} dx_1\dots dx_k\,g(x_1)\dots g(x_k) \big(\phi(x_1)^n \cdot_{T_{\la}}\dots \cdot_{T_{\la}}\phi(x_k)^n\big),
\end{align*}
and we see from the last expression that the $k$-th order term can be written as a sum over Feynman graphs with $k$ $n$-valent vertices labelled by $x_1, \dots, x_k$, without tadpoles, in such a way that to each internal line connecting vertices $x_j$ and $x_\ell$ there corresponds a factor $\Delta_{F,\la}(x_j-x_\ell)$, to each external line connected to the vertex $x_i$ there corresponds a factor $\phi(x_i)$ and to each vertex $x_\ell$ a factor $(-i)\int_{M} dx_\ell\,g(x_\ell)$. In particular we note that this reproduces the usual $S$ matrix in the $\la \to 0$ limit, in which $\Delta_{F,\la}$ reduces to the ordinary Feynman propagator.

More generally we can define, by the same formula, $S(A)$ for an arbitrary $A \in \cA$ as a formal power series with coefficients in $\cA$. Of course, the products $\star_\la$ and $\cdot_{T_\la}$ can be extended to the formal power series in the interaction in the obvious way. From now on, for simplicity, we will often refrain from specifying each time if the functionals we are considering are elements of $\cA$ or are formal power series with coefficients in $\cA$, and we will therefore use the symbol $\cA$ to denote both algebras.

Although Lorentz invariance and thus the ordinary causality properties are broken in the theory we are now considering, we observe that a remnant of causality is still valid.
In particular the following {\it temporal factorization property} holds for the time ordered products.
Given closed sets $C_1, C_2 \subset M$, we write $C_1 \gtrsim C_2$ if there exists a Cauchy surface $\Sigma_\tau = \{x \in M\,|\,t(x)=\tau\}$ (where $t : M \to \bR$ is such that $t(x)=x^0$) such that $C_1\subset J^{+}\Sigma_\tau$ and
$C_2 \subset J^{-}\Sigma_\tau$. Then, for $A, B \in \cA$ we write $A \gtrsim B$ if $\supp A \gtrsim \supp B$.
With this notation, if $A_i \gtrsim B_j$ for every $i,j$ then
\begin{equation}\label{eq:causal-factorisation-T}
A_1 \cdot_{T_\lambda}\dots \cdot_{T_\lambda} A_k \cdot_{T_\lambda}B_1 \cdot_{T_\lambda}\dots\cdot_{T_\lambda} B_l = (A_1\cdot_{T_\lambda}\dots\cdot_{T_\lambda} A_k) \star_\lambda ( B_1 \cdot_{T_\lambda}\dots \cdot_{T_\lambda} B_l). 
\end{equation}
  
The previous temporal factorization property implies a similar factorization property of the $S$ matrix. We have actually the following proposition.

\begin{proposition} \label{pr:temporal-factorisation}
For any $A,B,C \in \mathcal{A}$ such that $B\in \mathcal{A}_{\text{loc}}$ and $A \gtrsim C$ there holds
\[
S(A+B+C)=S(A+B)\star_\lambda S(B)^{-1}\star_\lambda S(B+C).
\]
\end{proposition}
\begin{proof}
If $A\gtrsim C$ the commutativity and the temporal factorization property \eqref{eq:causal-factorisation-T} of the time ordered products $\cdot_{T_\lambda}$  imply that 
\[\begin{split}
S(A+C)&=\sum_{k=0}^{+\infty} \frac{(-i)^k}{k!} (A+C)^{\cdot_{T_\la} k}= \sum_{k=0}^{+\infty}\sum_{j=0}^k \frac{(-i)^k}{j!(k-j)!} A ^{\cdot_{T_\la} j} \cdot_{T_\la} C^{\cdot_{T_\la}(k-j)}\\
&= \sum_{k=0}^{+\infty}\sum_{j=0}^k \frac{(-i)^k}{j!(k-j)!}  A ^{\cdot_{T_\la} j} \star_\la C^{\cdot_{T_\la}(k-j)} 
=S(A)\star_\lambda S(C).
\end{split}\]
Since $A\gtrsim C$, a Cauchy surface $\Sigma_\tau$ exists such that 
$\text{supp}(A)\subset J^+\Sigma_\tau$ and $\text{supp}(C)\subset J^{-}\Sigma_\tau$.
Now given $\epsilon > 0$, we fix functions $\chi_{\pm}, \chi_r \in C^\infty(\bR)$ such that $\chi_+ +\chi_r + \chi_- = 1$, $\supp\,\chi_\pm \subset \{t \in \bR\,|\, \pm(t-\tau)>0\}$ and $\supp\, \chi_r \subset (\tau-\epsilon, \tau + \epsilon)$. Then if
\[
B(\phi) = \sum_{s=0}^N \int_M dx\,g_s(x)\phi(x)^s,
\]
we define
\[
B_\pm(\phi) = \sum_{s=0}^N \int_M dx\,g_s(x)\chi_\pm(x^0)\phi(x)^{s}, \qquad B_r(\phi) = \sum_{s=0}^N \int_M dx\,g_s(x)\chi_r(x^0)\phi(x)^{s},
\]
so that $B = B_+ + B_r + B_-$ with $\text{supp}\, B_\pm\subset  J^\pm\Sigma_\tau$ and $\text{supp} \,B_r \subset
\Sigma_{\tau,\epsilon}$, where $\Sigma_{\tau,\epsilon}=\{x \in M\,|\,t(x) \in (\tau-\epsilon,\tau+\epsilon)\}$ is an $\epsilon-$neighbourhood of $\Sigma_\tau$. Therefore, $A+B_+ \gtrsim B_-+C$ and thus, if $B_\epsilon := B_+ + B_-$,
\[
S(A+B_\epsilon+C) = S(A+B_++B_-+C) = S(A+B_+)\star_\lambda S(B_-+C). 
\]
We now restore the missing parts of $B_\epsilon$ obtaining
\begin{align*}
S(A+B_\epsilon+C) &= S(A+B_+)\star_\lambda S(B_-)\star_\lambda S(B_-)^{-1}\star_\lambda S(B_+)^{-1}   \star_\lambda S(B_+)\star_\lambda  S(B_-+C)  
\\ &=
S(A+B_\epsilon)\star_\lambda S(B_\epsilon)^{-1}\star_\lambda S(B_\epsilon+C).
\end{align*}
To conclude the proof we take the limit $\epsilon \to 0$ of the previous relation.  
To this end, assume that $A$ is given by a finite sum of elements of the form~\eqref{eq:A}, and introduce the shorthand notations $X := (x_1, \dots, x_\ell) \in M^\ell$, $N:= (n_1,\dots,n_\ell) \in \bN^\ell$, $\Phi^{(N)}(X) := \prod_{i=1}^\ell \phi(x_i)^{n_i}$. Then the $k$-th order of $S(A+B_\epsilon)$ is a finite linear combination of terms of the form
\begin{multline*}
\int_{M^{\sum_h \ell_h+ k-j}} dX_1\dots dX_j dy_1\dots dy_{k-j} \prod_{h=1}^kf_{\ell_h}(X_h)\prod_{i=1}^{k-j}g_{s_i}(y_i)\chi_\epsilon(y^0_i)\times\\
 \Phi^{(N_1)}(X_1)\cdot_{T_\la}\dots \cdot_{T_\la}\Phi^{(N_k)}(X_k)\cdot_{T_\la} \phi(y_1)^{s_1}\cdot_{T_\la}\dots\cdot_{T_\la}\phi(y_{k-j})^{s_{k-j}},
\end{multline*}
where $\chi_\epsilon = \chi_++\chi_- = 1-\chi_r$. Therefore, since $\chi_\epsilon(t) \to 1$ for all $t \in \bR$ as $\epsilon \to 0$, since all the functions $f_\ell$, $g_s$ have compact support, and since the propagators $\Delta_{F,\lambda}$ appearing in the $\cdot_{T_\la}$ products are bounded, one sees that for each fixed $\phi \in \cC$ the above integral converges, as $\epsilon \to 0$, to the corresponding term in the expression of $S(A+B)$ by the dominated convergence theorem. A similar argument applies to $S(B_\epsilon)$, $S(A+B_\epsilon +C)$ and $S(B_\epsilon +C)$, thus concluding the proof. 
\end{proof}

We notice that both the free and interacting versions of the {\it Time slice axiom} \cite{FredenahgenChilian} hold also for these theories. 
The proof of this fact can be done in an analogous way as in \cite[Thm.\ 2 and Sec.\ 3]{FredenahgenChilian}. 
We recall that according to this axiom $(\mathcal{F}^p(\Sigma_{0,\epsilon}),\star_\lambda)$, is isomorphic to $(\mathcal{F}^p,\star_\lambda)$ up to elements vanishing on shell. Here 
$\mathcal{F}^p(\Sigma_{0,\epsilon})$ denotes the set of elements of $\mathcal{F}^p$ supported on $\Sigma_{0,\epsilon}$ which is an $\epsilon-$neighbourhood of the Cauchy surface at $t=0$.

To illustrate the basic idea of the time slice axiom in the case of free theories consider the solution of the free equation of motion given by the convolution $\psi_f=\Delta_\lambda * f$,  where $f$ is a compactly supported smooth function. According to the time slice axiom it is possible to find a compactly supported smooth function $g$ supported in $\Sigma_{0,\epsilon}$ such that $\psi_f=\psi_g$. Without loosing generality, we prove this fact assuming that $\text{supp}f\subset \{(t,\mathbf{x})| t<-\epsilon\}$. Now, to construct the function $g$ associated to $f$ we use a smooth function $\xi:\mathbb{R}\to [0,1]$ such that $\xi(t)=0$ for $t>\epsilon$ and $\xi(t)=1$ for $t<-\epsilon$. With this function at hand, we define $h:=\xi \Delta_R * f$ where $\Delta_R$ is the standard retarded fundamental solution associated to the differential operator $\Box-m^2$. Notice that $h$ is a compactly supported smooth function and it can be used to construct $g$ with the desired properties as $g:=f-(\Box-m^2)h$.   Actually it is clear that $\Delta_\lambda*(\Box-m^2)h = 0$ and
\[
(\Box-m^2)h = \xi (\Box-m^2)\Delta_R *f  + \ddot \xi \Delta_R* f  +2\dot\xi\partial_t \Delta_R* f
= f  + \ddot \xi \Delta_R* f  +2\dot\xi\partial_t \Delta_R *f,
\] 
hence the support of $g$ is contained in the region where $\xi$ is not constant, namely in $\Sigma_{0,\epsilon}$.
With similar arguments it is now possible to prove that a linear field smeared with $f$ is equal to a field smeared with $g$ up to another element which vanishes on shell. More generally, every element of $\mathcal{F}^p$ can be written as an element of $\mathcal{F}^p$ supported on $\Sigma_{0,\epsilon}$ plus an element of the ideal generated by the free equation of motion.    
For this reason in the construction of the interacting vacuum state performed in Sec.~\ref{sec:adiabatic},
we shall only consider observables which are in $\mathcal{A}(\Sigma_{0,\epsilon})$.
The state we will construct will not explicitly depend on $\epsilon$ and will be invariant under time translations. Knowing it on $\mathcal{A}(\Sigma_{0,\epsilon})$ we will be able to extend it to $\mathcal{A}(M)$ by time translation invariance.

It is interesting to notice that, thanks to the locality of $V_g$, the $S$ matrix thus defined is unitary.

\begin{proposition}
For $V_g$ as in~\eqref{eq:V} there holds
\[
S(V_g) \star_\la S(V_g)^* = 1 =S(V_g)^*\star_\la S(V_g)
\]
to all orders in perturbation theory.
\end{proposition}

\begin{proof}
Given points $x_1,\dots,x_k \in M$ and using the notations $X := \{x_1,\dots,x_k\}$, $T^k_\la(X) := \phi(x_1)^n \cdot_{T_\la}\dots \cdot_{T_\la} \phi(x_k)^n$, we define recursively functionals
\begin{equation}\label{eq:antiT}
\bar T^k_\la(X) := - \sum_{\substack{Y \subset X\\|Y| > 0}} (-1)^{|Y|}T_\la^{|Y|}(Y) \star_\la \bar T_\la^{|X\setminus Y|}(X\setminus Y).
\end{equation}
and $T^0_\la = 1$. One sees immediately, by induction, that $\bar T^k_\la$ is symmetric in its arguments, so it is actually only a function of the unordered points $X$. If we define then, as a formal series with coefficients in $\cA$,
\[
\tilde S(V_g) := \sum_{k=0}^{+\infty} \frac{i^k}{k!} \int_{M^k} dx_1\dots dx_k\,g(x_1)\dots g(x_k) \bar T^k_\la(x_1,\dots,x_k),
\]
we compute, using~\eqref{eq:antiT}, that $S(V_g)\star_\la \tilde S(V_g) = 1$. Moreover, from~\eqref{eq:antiT} one verifies, again by induction, that
\[
\bar T^k_\la(X) := - \sum_{\substack{Y \subset X\\|Y| < k}} (-1)^{|X\setminus Y|}\bar T_\la^{|Y|}(Y) \star_\la T_\la^{|X\setminus Y|}(X\setminus Y),
\]
and this in turn implies $\tilde S(V_g) \star_\la S(V_g) =1$, i.e., $\tilde S(V_g) = S(V_g)^{-1}$. If now the coupling function $g \in C^\infty_0(M)$ is of the form $g = g_1 + g_2$ with $\supp \,g_2 \gtrsim \supp\, g_1$, from Proposition~\ref{pr:temporal-factorisation} we get $S(V_g)^{-1} = S(V_{g_1})^{-1}\star_\la S(V_{g_2})^{-1}$ so that, identifying the corresponding orders,
\[\begin{split}
& \frac{i^k}{k!} \int_{M^k} dx_1\dots dx_k (g_1(x_1)+g_2(x_1))\dots(g_1(x_k)+g_2(x_k)) \bar T^k_\la(x_1,\dots, x_k)\\
&= \sum_{j=0}^k\frac{i^k}{j! (k-j)!} \int_{M^j} dx_1\dots dx_j \int_{M^{k-j}} dx_{j+1}\dots dx_k\, g_1(x_1)\dots g_1(x_j)g_2(x_{j+1})\dots g_2(x_k) \times\\
&\qquad\qquad\qquad\qquad\qquad\times \bar T^j_\la(x_1,\dots,x_j)\star_\la\bar T^{k-j}_\la(x_{j+1},\dots,x_k).
\end{split}\]
Grouping now, in the left hand side, the terms in $\prod_{h=1}^k (g_1(x_h)+g_2(x_h))$ according to the number $j$ of $g_1$'s, changing variables and taking into account the symmetry of $\bar T^k_\la$, we conclude that the $\bar T^k_\la$ are anti-chronological products, namely
\[
\bar T^k_\la(x_1,\dots,x_k) = \bar  T^j_\la(x_1, \dots, x_j)\star_\la \bar T^{k-j}_\la(x_{j+1},\dots,x_k)
\]
if $\{x_{j+1},\dots,x_k\} \gtrsim \{x_1,\dots,x_j\}$. Therefore, if in particular $x^0_{j_1} \leq x^0_{j_2} \leq \dots \leq x^0_{j_n}$,
\[\begin{split}
\bar T^n_\la(x_1,\dots,x_n) &= \bar T^1_\la(x_{j_1}) \star_\la \dots \star_\la\bar T^n_\la(x_{j_n}) = \phi(x_{j_1})^n \star_\la \dots \star_\la\phi(x_{j_n})^n \\
&= [\phi(x_1)^n \cdot_{T_{\la}} \dots \cdot_{T_{\la}}\phi(x_n)^n]^*,
\end{split}\]
which entails $S(V_g)^{-1} = S(V_g)^*$.
\end{proof}

Summarizing, we see that in this approach the $S$ matrix, with a fixed infrared cutoff, is unitary and automatically ultraviolet finite, without the need of renormalization, as a consequence of the boundedness of the modified Feynman propagator $\Delta_{F,\la}$. It is worth observing that the proof of ultraviolet finiteness is here much easier than the corresponding one in~\cite{BDFP03}. It is also interesting to compare~\eqref{eq:modFeyn} with the propagator obtained by defining QFT on QST through the Filk rules \cite{Filk}:
\[
\frac{i e^{-\la^2 \langle p\rangle^2}}{p^2+m^2-i\epsilon}
\]
which results in a nonunitary $S$ matrix.

We also notice that the nice properties satisfied by $\Delta_{+,\lambda}$ and $\Delta_{F,\lambda}$ stated in Proposition \ref{pr:prop} and \ref{pr:prop1} imply that no divergences occur
in terms like
\[
(V_g \star_{\lambda} V_g)  \cdot_{T_\lambda} V_g
\] 
and thus the Bogoliubov map~\eqref{eq:Bog} can be applied to every element of $\mathcal{A}$ and not only to time ordered products of local functionals. 

Furthermore, in this case, the Bogoliubov map can also be inverted in the sense of perturbation theory. 
Actually, since $S(V_g)\cdot_{T_\lambda} S{(-V_g)} = 1$, for $A, B \in \cA$ it holds that
\[
 R_{V_g}(A) = S(V_g)^{-1}\star_\lambda (S(V_g) \cdot_{T_\lambda} A) = B \quad \Leftrightarrow \quad   A =  S({-V_g}) \cdot_{T_\lambda} (S(V_g) \star B) = R_{V_g}^{-1}(B).
\]
Contrary to the ordinary case, we have thus access to the interacting product in $\mathcal{A}$ defined as 
\[
F_1\star_{\lambda, V_g} F_2 = R_{V_g}^{-1}\left(R_{V_g}(F_1)\star_\lambda R_{V_g}(F_2)\right), \qquad F_i\in\mathcal{A}
\]
In this way we obtain that the Bogoliubov map is a $*-$automorphisms of algebras
\[
R_{V_g}:(\mathcal{A},\star_\lambda) \to (\mathcal{A},\star_{\lambda,V_g}).
\]
This implies in particular that results of~\cite{HR} can be applied.

We put on record here the following result on the structure of the perturbative expansion of the Bogoliubov map, which we will use later on.
\begin{proposition}\label{prop:bogoexpansion}
Let $A \in \cA$ be of the form~\eqref{eq:A}. Then the $k$-th perturbative order of $R_{V_g}(A)$ is a finite linear combination of terms of the form
\begin{multline}\label{eq:bogok}
\int_{M^{k+\ell}} dx_1\dots dx_k dy_1\dots dy_\ell \prod_{j=1}^k g(x_j) f(y_1,\dots,y_\ell) \times \\ 
\prod_{l \in E(G)} \Delta_l(z_{r(l)}-z_{s(l)}) \Phi(x_1,\dots,x_k,y_1,\dots,y_\ell),
\end{multline}
where: $G$ is a graph with vertices $V(G)$ satisfying $\{1,\dots,k\} \subset V(G) \subset \{1,\dots,k+\ell\}$ and such that each of its connected components contains a vertex in the set $\{k+1,\dots,k+\ell\}$; $\Delta_l$ can be either $\Delta_{F,\la}$ or $\Delta_{+,\la}$; $(z_1,\dots,z_{k+\ell}) = (x_1,\dots,y_\ell)$; $r, s : E(G) \to V(G)$ are the range and source maps of the graph $G$; and $\Phi$ is a monomial in the fields $\phi(x_1),\dots ,\phi(y_\ell)$.
\end{proposition}
\begin{proof}
Assume that $A_p$ is a formal power series with lowest order $p$, whose coefficients can be written as finite linear combinations of terms of the form~\eqref{eq:bogok} (with $k \geq p$). Then
\begin{align*}
R_{V_g}(A_p) &= S(V_g)^{-1}\star_\la (S(V_g)\cdot_{T_\la} A_p )\\
&=  S(V_g)^{-1}\star_\la S(V_g) \star_\la A_p + 
S(V_g)^{-1}\star_\la \left( S(V_g) \cdot_{T_\la} A_p  - S(V_g) \star_\la A_k \right)\\
&=A_p+R_{V_g}(A_{p+1}),
\end{align*}
where $A_{p+1} := \sum_{r= 1}^{+\infty}S(-V_g)\cdot_{T_\la} \langle S(V_g)^{(r)}, (\Delta_{F,\la}^{\otimes r}-\Delta_{+,\la}^{\otimes r})  A_p^{(r)}\rangle$. Then since $S(V_g)' = S(V_g)\cdot_{T_\la} V'_g$, one has that for $r \geq 1$, $S(V_g)^{(r)} = S(V_g) \cdot_{T_\la} W_r$ where $W_r$ is a $\cdot_{T_\la}$-polynomial in $V'_g,\dots, V^{(r)}_g$ of order at least one. Hence $A_{p+1} = \sum_{r= 1}^{+\infty}\langle W_r, (\Delta_{F,\la}^{\otimes r}-\Delta_{+,\la}^{\otimes r})  A_p^{(r)}\rangle$ is again a formal power series with lowest order $p+1$, whose coefficients can be written as finite linear combinations of terms of the form~\eqref{eq:bogok} with $k \geq p+1$. Applying now the above decomposition recursively starting from $A_0 := A$, one gets $R_{V_g}(A) = \sum_{p=0}^{+\infty} A_p$ where the sum is finite order by order, and thus we obtain the claim.
\end{proof}

\section{Adiabatic limits}\label{sec:adiabatic}

In this Section we analyze the adiabatic limits of the theory discussed above.
We shall here consider an interaction Lagrangian $V_g \in \cA_{\text{loc}}$ of the form~\eqref{eq:V} and we want to study the limit where $g\to1$ of expectation values of observables in a state.
In particular, we restrict our attention to observables supported in 
\[
\Sigma_{0,\epsilon} := \{x\in M| t(x) \in (-\epsilon,\epsilon)  \}.
\]
Then on the algebra $\cA(\Sigma_{0,\epsilon})$ the temporal factorization property stated in Proposition \ref{pr:temporal-factorisation} implies that we might restrict the cutoff of the interaction. To make this precise, we introduce the relative $S$ matrix
\[
S_{V_g}(A) := S(V_g)^{-1}\star_\la S(V_g+A), \qquad A \in \cA.
\]
We observe now that

\begin{proposition}\label{pr:cutoff-unitari}
Given $A\in\cA(\Sigma_{0,\epsilon})$ and two compactly supported smooth functions $g,g'$, we have that 
if $\supp\,(g-g') \subset  J^{+}(\Sigma_{\epsilon})$ then
\begin{equation}\label{eq:prima-factorisation}
S_{V_g}(A) = S_{V_{g'}}(A), 
\end{equation}
and if $\supp\,(g-g') \subset  J^{-}(\Sigma_{-\epsilon})$  then
\begin{equation}\label{eq:seconda-factorisation}
S_{V_g}(A) = W_{g,g'}^{-1}\star_\lambda   S_{V_{g'}}(A) \star_\lambda  W_{g,g'}, 
\end{equation}
where $W_{g,g'} = S(V_{g'})^{-1}\star_\lambda S({V_g})$ is (formally) a unitary element of $\cA$.
\end{proposition}
\begin{proof}
The interaction Lagrangian $V_g$ is linear in $g$:
\[
V_g = V_{g'} + V_{g-g'}.
\]
Then, if $\supp\,(g-g') \subset  J^{+}(\Sigma_{\epsilon})$ it holds that $V_{g-g'} \gtrsim A$. 
In view of Proposition \ref{pr:temporal-factorisation} we have that
\[
S(V_g+A) =  S(V_{g-g'} + V_{g'} +A) =
S(V_{g-g'} + V_{g'})\star_\lambda
S(V_{g'})^{-1}\star_\lambda
S(V_{g'} +A)
\]
This implies \eqref{eq:prima-factorisation}. 
If now $\supp\,(g-g') \subset  J^{-}(\Sigma_{-\epsilon})$ it holds that $A  \gtrsim V_{g-g'}$ hence Proposition \ref{pr:temporal-factorisation} implies that 
\[
S(V_g+A) =  S(A + V_{g'} +V_{g-g'}) =
S(V_{g'} +A)\star_\lambda
S(V_{g'})^{-1}\star_\lambda
S(V_{g'} +V_{g-g'}) = S(V_{g'} +A) \star_\lambda  W_{g,g'} 
\]
This implies \eqref{eq:seconda-factorisation} thus concluding the proof.
\end{proof}

This proposition allows us to specialize the form of the cutoff $g$ as $g(t,\mathbf{x}) =\chi(t) h(\mathbf{x})$ where $\chi \in C^\infty_0(\bR)$ is a time cutoff and $h \in C^\infty_0(\Sigma_0)$ a spatial cutoff. With this choice, the interaction potential will be denoted by
\[
V_{\chi,h}= \int_{M} \chi(t)h(\mathbf{x})\phi^n(x) dx.
\]
Actually, there holds clearly, as formal power series,
\[
R_{V_g}(A) = i\left.\frac{d}{ds} S_{V_g}(s A) \right|_{s=0},
\]
and therefore, in the procedure of taking the adiabatic limit for elements of the form $R_{V_{\chi,h}}(A)$ with  $A\in \mathcal{A}(\Sigma_{0,\epsilon})$, Proposition  \ref{pr:cutoff-unitari} implies that the form of $\chi$ in $(\epsilon,+\infty)$ is irrelevant. Furthermore, if we modify $\chi$ in $(-\infty,-\epsilon)$ we obtain the same object up to the adjoint action of a unitary element  that will not affect the existence of the adiabatic limit (see the remark following Corollary~\ref{cor:adiabatic}). 
For this reason we shall fix $\chi$ to be a smooth function equal to $1$ on $(-\epsilon,\epsilon)$ and supported on 
$(-2\epsilon,2\epsilon)$ once and for all and we shall care only about the limit where $h\to 1$. This choice will be useful to control the convergence of the adiabatic limit up to unitary equivalence.

We now introduce the free vacuum state  $\omega_\lambda$ on $(\mathcal{A},\star_\lambda)$, given by evaluation at $\phi = 0$, namely
\[
\omega_\lambda(F) := F(0).
\]
We want to show that  the evaluation of $\omega_{\lambda}\circ R_{V_{\chi,h}}$ on   $\mathcal{A}(O)$, where $O \subset \Sigma_{0,\epsilon}$ is a bounded region, converges in the limit $h\to 1$. 

We have actually the following theorem.

\begin{theorem}\label{thm:adiabatic}
Let $A$ be an element of $\mathcal{A}(O)$ where $O\subset \Sigma_{0,\epsilon}$ is a bounded open set. 
Denoting by $R_{V_{\chi,h}}(A)^{[k]}$ the $k-$th order in the coupling constant of the formal power series defining $R_{V_{\chi,h}}(A)$,
the limit
\[
\omega_\lambda^{[k]}(A) := \lim_{h\to 1} \omega_\lambda\left(R_{V_{\chi,h}}(A)^{[k]}\right)
\]
is finite for all $k \in \bN$.  
\end{theorem}
\begin{proof}
Since $R_{V_{\chi,h}}(A)$ is linear in $A$, without loosing generality we assume that
\[
A(\phi) = \int_{M^\ell} dy_1\dots dy_\ell f(y_1,\dots ,y_\ell)\phi^{n_1}(y_1)\dots \phi^{n_\ell}(y_\ell), 
\]
where $f$ is a compactly supported continuous function whose support is in $O^\ell\subset \Sigma_{0,\epsilon}^\ell$.
By Prop.~\ref{prop:bogoexpansion} we can now expand 
\[\begin{split}
\omega_\lambda\left(R_{V_{\chi,h}}(A)^{[k]}\right) 
= \int_{M^{k+\ell}} d{x}_1\dots d{{x}_k} dy_1\dots dy_\ell\,&\chi(x^0_1)h(\mathbf{x}_1) \dots \chi(x^0_k)h(\mathbf{x}_k)\times\\
&D_\la({x}_1,\dots, {x}_k,{y}_1,\dots,{y}_\ell)  f({y}_1,\dots,{y}_\ell)
\end{split}\]
where $D_\la$ is a combination of propagators $\Delta_{F,\lambda}$,  $\Delta_{+,\lambda}$
which is actually a continuous function as discussed in Proposition \ref{pr:prop} and  Proposition \ref{pr:prop1}, and where the integral extends over a compact set, because of the support properties of $\chi$, $h$, $f$.
Moreover $D_\la$ can be written in the standard way as a finite sum of contributions labelled by graphs with $k + \ell$ vertices labelled by the points in the collection 
\[
{Z} := ({z}_1,\dots,{z}_{k+\ell})
=   ({x}_1,\dots, {x}_k,{y}_1,\dots, {y}_\ell) =: ({X},{Y}) \in 
\Sigma_{0,2\epsilon}^{k} \times O^\ell\subset M^{k+\ell},
\]
each pair of vertices being joined by a line labelled by either $\Delta_{F,\lambda}$ or $\Delta_{+,\lambda}$, and only graphs such that each of their connected components contains at least one vertex $y_j$ appear. 
We denote by $\cG$ the set of such graphs.
Consider now the contribution $D_{\la,G}({X},{Y})$ to $D_\la({X},{Y})$ labelled by $G\in\mathcal{G}$. 
In order to establish the required statement, it is sufficient, by the dominated convergence theorem and the uniform boundedness of $\chi$, $h$, $f$, to prove that $D_{\la,G}({X},{Y})$ is absolutely integrable on $\Sigma_{0,2\epsilon}^{k}\times O^\ell$.

To this end, we first notice that we just need to care about the integration over the space components $(\mathbf{X},\mathbf{Y})\in \Sigma_0^{k+\ell}$ of $(X,Y)$, because the integration in the time components is restricted to the compact set $[-2\epsilon,2\epsilon]^k\times[-\epsilon,\epsilon]^\ell$.  Then, thanks to the estimates given in Proposition \ref{pr:clustering} in Appendix~\ref{sec:bounds}
we have that
\[
|D_{\la,G}({X},{Y})| \leq \gamma_G(\mathbf{X},\mathbf{Y})
\]
uniformly in the time components, where  
\[
\gamma_G(\mathbf{X},\mathbf{Y}) = \prod_{l\in E(G)} C^{|E(G)|} e^{-m |\mathbf{z}_{s(l)}-\mathbf{z}_{r(l)}|}.
\]
with $E(G)$ denoting the set of edges in $G$ and, for every $l\in E(G)$, with $s(l)$ and $r(l)$ denoting its source and range.  In order to show that $\gamma_G$ is integrable we start by observing that since $G$ is a union of connected components, each of which contains at least one vertex labelled by some $y_j$, it
possesses at least one subgraph $G'$ which is a disjoint union of rooted trees, with roots in some of the points $Y$ and connecting all the elements of $\mathbf{X}$, and then
\[
\gamma_G \leq  C^{e} \gamma_{G'},
\]
where $e = |E(G)|-|E(G')|$.
To prove that $\gamma_{G'}$ is integrable we use the following recursive procedure. 
Let $L_x$, $L_y$ be the set of leaves of $G'$ labelled by elements of $\mathbf{X}$, $\mathbf{Y}$ respectively, and let $\Lambda \subset \Sigma_0$ be a compact set containing the projection of $O \subset M$. Since
\[
\int_{\bR^3} e^{-m|\mathbf{x}-\mathbf{z}|} d\mathbf{x} =\frac{8\pi}{m^3} \qquad \forall \,\mathbf{z}\in \bR^3,
\] 
we can estimate
\[
\int_{\bR^{3|L_x|}}\prod_{\mathbf{x}_j \in L_x} d\mathbf{x}_j \int_{\Lambda^{|L_y|}}\prod_{\mathbf{y}_h \in L_y}d\mathbf{y}_h\,\gamma_{G'}(\mathbf{X},\mathbf{Y}) \leq C^{|L_x|+|L_y|}\left(\frac{8\pi}{m^3}\right)^{|L_x|}|\Lambda|^{|L_y|} \gamma_{G''}(\mathbf{X}'',\mathbf{Y}''),  
\]
where $G''$ is the union of rooted trees obtained by pruning the leafs of $G'$, and $(\mathbf{X}'',\mathbf{Y}'')$ the labels of its vertices.
We can now iterate the same procedure until the graph is reduced to the set of roots of $G'$, which are integrated on $\Lambda$.
\end{proof}

The state obtained so far after the limit $h\to1$ depends on $\chi$, hence, it is in particular not invariant under time translation and thus cannot be the vacuum of the interacting theory. To address this problem, we start observing that states constructed by different $\chi$ can obtained up to formal unitary equivalence. Actually,
the result of Theorem~\ref{thm:adiabatic} can be easily extended to prove the existence of the adiabatic limit for suitable families of states on $(\cA,\star_\la)$ indexed by the spatial cutoff function $h$. To make this precise, we introduce the states
\begin{equation}\label{eq:omegaB}
\omega_\la^B(A) := \frac{\omega_\la(B^*\star_\la A \star_\la B)}{\omega_\la(B^*\star_\la B)}, \qquad A \in \cA,
\end{equation}
for $B \in \cA$. Moreover, for a function $f : M^k \to \bC$ we define its temporal support as the union of the supports of the functions $t \in \bR \mapsto f(t,\mathbf{x}_1,\dots,t,\mathbf{x}_k)$ as $(\mathbf{x}_1,\dots,\mathbf{x}_k) \in \Sigma_0^k$.
\begin{corollary}\label{cor:adiabatic}
Assume that $B_h \in \cA$ is a finite sum of functionals of the form
\[
\int_{M^s} dw_1\dots dw_s\, b(w_1,\dots,w_s) h(\mathbf{w}_1)\dots h(\mathbf{w}_s)\phi(w_1)^{r_1}\dots\phi(w_s)^{r_s}
\]
with $b$ a bounded continuous function on $M$ constructed with propagators $\Delta_{+\lambda}$ or $\Delta_{F\lambda}$ in such a way that they are related with the edges of a connected graph with $s$ vertices, with compact temporal support, and assume that
\[
\supp\, h \subset \supp\, h' \quad \Rightarrow \quad \omega_\la(B_h^* \star_\la B_h) \leq \omega_\la(B_{h'}^*\star_\la B_{h'}).  
\] 
Then for any $A \in \cA(O)$ with $O \subset \Sigma_{0,\epsilon}$ a bounded open set, the limit
\[
\lim_{h \to 1} \omega_\la^{B_h}(R_{V_{\chi,h}}(A)^{[k]})
\]
is finite for all $k \in \bN$.
\end{corollary}

\begin{proof}
Considering the diagrammatic expansion, similar to the one discussed in the proof of Theorem \ref{thm:adiabatic} we have that $\omega_\la(B_h^*\star_\la R_{V_{\chi,h}}(A)^{[k]}\star_\la B_h)$ is the sum of a ``disconnected'' term $\omega_\la(B_h^* \star_\la B_h) \omega_\la(R_{V_{\chi,h}}(A)^{[k]})$, coming from the graphs in which the vertices pertaining to $B$ and $B^*$ are only connected among themselves, and a ``connected'' term, which is a sum of integrals like
\[
\int_{M^{k+\ell+s+s'}} dXdYdWdW'\,\overline{b'(W')}b(W)f(Y)\prod_{j=1}^k \chi(x_j^0)h(\mathbf{x}_j) \prod_{i=1}^s h(\mathbf{w}_i) \prod_{p=1}^{s'} h(\mathbf{w}'_p)\tilde D_\la(X,Y,W,W'),
\]
where $W = (w_1,\dots,w_s)$, $W'=(w'_1,\dots,w'_{s'})$ and where $\tilde D_\la$ is a sum of connected graphs with vertices labelled by $(X,Y,W,W')$ and lines labelled by either $\Delta_{F,\la}$ or $\Delta_{+,\lambda}$. It is then clear that the limit $h \to 1$ of the latter integral exists by the same recursive argument used in the proof of Theorem~\ref{thm:adiabatic}. One then gets the desired statement taking into account the existence in $(0,+\infty]$ of $\lim_{h \to 1} \omega_\la(B_h^* \star_\la B_h)$ by monotonicity.
\end{proof}

The previous result entails in particular that if we change that temporal cutoff function $\chi$ to a $\chi' \in C^\infty_0(\bR)$ such that $\supp\, \chi' \subset (-\infty, 2\epsilon)$ and $\supp(\chi-\chi') \subset (-\infty, -\epsilon)$, for $A \in \cA(O)$ with $O \subset \Sigma_{0,\epsilon}$ bounded there exists
\[
\lim_{h \to 1} \omega_\la(R_{V_{\chi',h}}(A)) = \lim_{h \to 1} \omega_\la^{W_{g',g}}(R_{V_{\chi,h}}(A)),
\]
with $W_{g',g}$ the unitary element defined in Proposition~\ref{pr:cutoff-unitari} for $g(t,\mathbf{x}) = \chi(t)h(\mathbf{x})$ and $g'(t,\mathbf{x}) = \chi'(t)h(\mathbf{x})$, which clearly satisfies the hypotheses of the Corollary.
\medskip

As discussed above, the state obtained considering the limit $h\to1$ in Theorem \ref{thm:adiabatic} is defined on $\mathcal{A}(\Sigma_{0,\epsilon})$
 depends on $\chi$ and in particular it is not invariant under time translations.
For this reason it cannot be the ground state of the theory.
We discuss how to modify it in order to get a vacuum of the interacting theory. 
We shall in particular translate back in time the region where interaction is switched on.
In the limit where the interaction starts at past infinity we obtain the vacuum of the theory which is a state defined on $\mathcal{A}(M)$ and it is invariant under time translation.
Unfortunately, the direct analysis of this limit is not available. Nevertheless, it turns out to be easier to compute it in states which are invariant under the interacting time evolution. For this reason, we shall analyze this limit for interacting equilibrium states at finite temperature and we will then compute the limit where the inverse temperature $\beta$ tends to infinity.

The construction of equilibrium states at finite temperature for interacting fields in the adiabatic limit has been recently proposed by Fredenhagen and Lindner \cite{FredenhagenLindner}.
The main idea of the latter paper is to adapt to the framework of pAQFT the construction, by Araki, of KMS states of a C*-dynamical system in which the dynamics is obtained by perturbing the generator of a given reference (free) dynamics by an element of the C*-algebra (the interaction hamiltonian). In this way, Fredenhagen and Lindner obtain KMS states for the perturbative interacting theory in the presence of an adiabatic cutoff of the form $g(t,\mathbf{x}) = \chi(t)h(\mathbf{x})$, where $\chi$ is unity on a fixed time slice. These are then shown to be independent of the form of $\chi$ 
thanks to the causal factorization property of the $S$ matrix and to the KMS condition. Finally the existence of the adiabatic limit $h \to 1$ can be reduced to suitable clustering properties of the free KMS state.

In order to adapt the above ideas to the present framework, we start by observing that in view of the temporal factorization property stated in Proposition \ref{pr:temporal-factorisation} and thanks to Proposition \ref{pr:cutoff-unitari}, as soon as the interacting field $A$ is supported in $\mathcal{O}\subset \Sigma_{0,\epsilon} $ we can modify $\chi$ in the future of $\Sigma_\epsilon$ without altering the expectation value of $R_{V_{\chi,h}}(A)$. Hence, we assume now to have $\chi=1$ in the interval $[-\epsilon, T]$ for some large $T > \epsilon$ and $\chi = 0$ outside $(-2\epsilon, T+\epsilon)$.
Notice that, in this way, $R_{V_{\chi,h}}(A)$ is actually independent of $T$ for $A \in \cA(\Sigma_{0,\epsilon})$, so we may think that, morally, $\chi = 1$ in the future of $\Sigma_{-\epsilon}$ and $\chi=0$ in the past of $\Sigma_{-2\epsilon}$, even if for such a $\chi$ the interaction $V_{\chi, h} = \int_M dx \chi(t)h(\mathbf{x}) \phi(x)^n$ is not a well defined element of $\cA$ because the integral may not converge for some $\phi \in \cC$. In particular, $R_{V_{\chi,h}}(A) \in \cA(\Sigma_{0,2\epsilon})$ for $A \in \cA(\Sigma_{0,\epsilon})$. In the following, when it does not cause confusion, we will often use the simplified notation $V$ for $V_{\chi,h}$.

To obtain a time invariant state we have to perform a time translation to minus infinity of the cutoff function $\chi$. 
To make this point precise, we have to discuss the form of the free and interacting time translation. 
Let $A \in \cA$ by any field observable, we write $A_t[\phi]=A[\phi_t]$ where $\phi_t(x) =\phi(x+t e_0)$.
The free time evolution, given by
\[
\alpha_t(A) := A_t, \qquad t \in \bR,
\]
defines a group of $*-$automorphisms of the algebra of observables thanks to the translation invariance of $\Delta_{+,\la}$.
On the other hand, the interacting time evolution $\alpha_t^V$ is defined by
\[
\alpha_t^V(R_V(A)) := R_V(A_t), \qquad t \in \bR.
\]
The relationship between the free and interacting time evolution can be expressed as follows.

\begin{proposition}
There exists unitary elements $U(t) \in \cA$, $t \in \bR$, such that
\[
\alpha_t^V(R_V(A)) =  U(t)\star_\la \alpha_t(R_V(A))\star_\la U(t)^{*}, \qquad t \in \bR.
\]
Moreover, $t \mapsto U(t)$ is a cocycle under the free time evolution:
\[
U(t+s) = U(t)\star_\la\alpha_t U(s), \qquad t, s \in \bR
\]
\end{proposition} 

\begin{proof}
Since
 $V$ is not invariant under time translations because of the cutoff $\chi$, we have that  
\begin{equation*}
\alpha_t^V(R_V(A)) = R_V(A_t) 
=  i\left.\frac{d}{ds} S_V(sA_t) \right|_{s=0} 
\end{equation*}
and choosing now $T > t-\epsilon$ so large that $A \in \cA((-\epsilon,T-t)\times \Sigma_0)$, we can write $V - V_t = V_t^+ + V_t^-$ with $V_t^\pm$ defined by the temporal cutoffs $s \mapsto [\chi(s) - \chi(s-t)]\theta(\pm s)$; this implies $\supp V^+_t \gtrsim \supp A_t \gtrsim \supp V_t^-$ and therefore by the temporal factorization property given in Proposition \ref{pr:temporal-factorisation},
\begin{align*}
S_V(sA_t) &=S(V)^{-1}\star_\la S(V+sA_t) =  S(V)^{-1}\star_\la S(V_t+sA_t+V-V_t) \\
&=S(V)^{-1}\star_\la S(V_t+sA_t+V^+_t+V^-_t)\\
&=S(V)^{-1} \star_\la S(V_t^++V_t) \star_\la S(V_t)^{-1}\star_\la S(sA_t+V_t^-+V_t)\\
&=S(V)^{-1} \star_\la S(V_t^++V_t) \star_\la S(V_t)^{-1}\star_\la S(V_t+sA_t)\star_\la S(V_t)^{-1}\star_\la S(V_t^-+V_t)\\
&=S(V)^{-1} \star_\la S(V_t^++V_t) \star_\la \alpha_t(S_V(sA))\star_\la S(V_t)^{-1}\star_\la S(V_t^-+V_t)
\end{align*}
where in the fourth equality we used the fact that $sA_t+V^-_t$ is supported in the past of $V^+_t$, in the fifth equality the fact that $V^-_t$ is supported in the past of $A_t$ and in the last equality the translation invariance of the Feynman propagator $\Delta_{F,\la}$. Moreover we observe that, again by temporal factorization,
\[\begin{split}
S(V-V^-_t)\star_\la S(V_t)^{-1}\star_\la S(V^-_t+V_t) &= S(V_t+V^+_t)\star_\la S(V_t)^{-1}\star_\la S(V^-_t+V_t) \\
&= S(V^+_t+V_t+V^-_t) = S(V)
\end{split}\]
and therefore
\[
S_V(sA_t) = S(V)^{-1}\star_\la  S(V-V^-_t) \star_\la \alpha_t(S_V(sA))\star_\la [S(V)^{-1}\star_\la  S(V-V^-_t)]^{-1}.
\]
Now taking the derivative with respect to $s$ at $s=0$ and using  the fact that $S(V)$ is unitary for any local interaction Lagrangian $V$, we obtain the statement with $U(t)= S(V)^*\star_\la S(V-V^-_t)$.
\end{proof}

Moreover, it is also easy to see, again by temporal factorization, that $U(t)$ defined in the previous proposition is independent of $T > t-\epsilon$. Since $V_t^+ \to 0$ as $T \to +\infty$, this means that we may also think of $U(t)$ as $S(V)^{-1}\star_\la S(V_t)$, with a temporal cutoff such that $\chi = 1$ in the future of $\Sigma_{-\epsilon}$. 

To construct the vacuum of the interacting theory, we keep for now $h$ of compact support and we consider the limit   
\begin{align*}
\lim_{t\to\infty} \omega_\la(\alpha_t^V R_{V}(A)) &=  \lim_{t\to\infty} \omega_\la(  U(t) \star   \alpha_t  R_{V}(A)\star U(t)^{*})
\\
&=  \lim_{t\to\infty} \omega_\la(  U(-t)^{*}\star   R_{V}(A)\star U(-t))
\\
&=  \lim_{t\to\infty} \omega_\la^{U(-t)}(R_{V}(A))
\end{align*}
where we have used the invariance under the free time translations of $\omega_\la$ and the definition \eqref{eq:omegaB}.

In order to show that the above limit actually exists, as a first step we observe that the following expansion of $\omega_\la^{U(t)}$ holds
\begin{equation}\label{eq:expansion-commutators}
\omega_\la^{U(-t)}(A)=\sum_{n\geq 0} i^n\int_{t S_n} dT \,\omega_\la \left([K_{-t_1}[K_{-t_2} \dots[K_{-t_n},A]\dots ]]\right)
\end{equation}
where $T=(t_1,\dots, t_n)$ and the domain of integration $tS_n$ is such that $0<t_n<\dots <t_1<t$. 
Furthermore, $K$ is the generator of the cocycle $U(t)$ and $K_t=\alpha_t(K)$. As shown in \cite{FredenhagenLindner} in a similar context, we have that 
\[
K=-i\left.\frac{d}{dt}U(t) \right|_{t=0},\qquad K = R_V{\dot{V}},  
\]
where $\dot{V}:=V_{\dot{\chi}^-,h}$ with $\dot \chi^-(s) := \dot \chi(s) \theta(-s)$. Hence, $\dot{V}$ is supported in the past of $\Sigma_{0,\epsilon}$ because $\dot{\chi}^-$ is supported in the interval $(-2\epsilon,-\epsilon)$; also, $K$ is supported in $\Sigma_{0,2\epsilon}$. 

A further essential ingredient are the connected correlation functions of $\omega_\la$, defined recursively by
\[
\omega_\la^c(A_1\otimes \dots \otimes A_n) := \omega_\la(A_1\star_\la \dots \star_\la A_n) - \sum_{P} \prod_{I \in P} \omega_\la^c\big(\otimes_{i \in I} A_i\big), \qquad A_1,\dots,A_n\in \cA,
\]
where the sum runs over all partitions $P$ of the set $\{1,2,\dots,n\}$ into at least two nonvoid subsets. 
The analiticity properties of $\Delta_{+\lambda}$ stated in Proposition \ref{pr:clustering} imply that the function
$(t_1,\dots, t_n)\mapsto\omega_\la^c(\alpha_{t_1}A_1\otimes \dots \otimes \alpha_{t_n}A_{t_n})$ can be analitically continued to $\text{Im}(t_1) < \dots <\text{Im}(t_n)$. 

To control the limit $h\to 1$ it is useful to exploit the interplay of the expansion \eqref{eq:expansion-commutators} with a similar expansion which holds in the domain of imaginary times for interacting KMS states, see e.g. Theorem 2 in \cite{BKR} or Proposition 3 in \cite{FredenhagenLindner}.

To present this connection we observe that a KMS state on $\mathcal{A}$ at inverse temperature $\beta$ with respect to free time translations $\alpha_t$ can be defined by
\begin{equation}\label{eq:free-KMS-state}
\omega_{\la,\beta}(A) := \omega_\la( e^{\frac 1 2\int dx dy [\Delta_{\beta,\la}(x-y)-\Delta_{+,\la}(x-y) ]\frac{\delta^2}{\delta\phi(x)\delta\phi(y)}} A), \qquad A \in \cA,
\end{equation}
with the two-point function
\begin{equation}\label{eq:two-point-KMS}
\Delta_{\beta,\lambda}(t,\mathbf{x}) 
:= 
\frac{e^{-\lambda^2 m^2}}{(2\pi)^3} \int_{\mathbb{R}^3}  \frac{ e^{-2\lambda^2 |\mathbf{p}|^2} e^{i\mathbf{p}\mathbf{x}}}{2\sqrt{|\mathbf{p}|^2+m^2}}
\left(\frac{e^{-it \sqrt{|\mathbf{p}|^2+m^2}}}{1-e^{-\beta\sqrt{|\mathbf{p}|^2+m^2}}}  - \frac{e^{it \sqrt{|\mathbf{p}|^2+m^2}}}{1-e^{\beta\sqrt{|\mathbf{p}|^2+m^2}}}    \right)   d\mathbf{p}.
\end{equation}
Whenever $h$ is of compact support, a KMS state for the interacting theory is obtained modifying $\omega_{\la,\beta}$ in the following way. 
Following \cite{FredenhagenLindner}, we observe that the function $t\mapsto \omega_{\beta,\lambda}^{U(-t)}(A)$, for $A\in\mathcal{A}$, expanded as in \eqref{eq:expansion-commutators},
 can be extended to a bounded continuous function on $\text{Im}\,t\in [-\frac{\beta}{2},0]$ which is analytic for $\text{Im}\,t \in (-\frac{\beta}{2},0]$. 
 The functional over $\mathcal{A}$ obtained for $\text{Im}\,t=-\frac{\beta}{2}$ and denoted by $\omega_{\beta,\la}^{U(i\frac{\beta}{2})}$  is a KMS state at inverse temperature $\beta$ with respect to $\alpha_t^V$. 
Notice that this state can be formally written as
\begin{equation}\label{eq:formal-KMS-state}
\omega_{\beta,\la}^{U(i \frac{\beta}{2})}(A)  = \frac{ \omega_{\beta,\la} (U(i \frac{\beta}{2})^*\star_\la A\star_\la U(i \frac{\beta}{2}))}{
\omega_{\beta,\la} (U(i \frac{\beta}{2})^*\star_\la  U(i \frac{\beta}{2}))}
\end{equation}
and it is actually expanded as follows\footnote{In the paper \cite{FredenhagenLindner} and in particular in Proposition 4 the KMS condition is used to rewrite  
$\omega_{\beta,\la} (U(i \frac{\beta}{2})^*\star_\la A\star_\la U(i \frac{\beta}{2}) )= \omega_{\beta,\la} (A\star_\la U(i \beta) )$ where the equality holds for the expansions of the form \eqref{eq:expansion-beta} computed on both sides. Here we are considering $\omega_{\beta,\la}^{U(i \frac{\beta}{2})}$ because its expansion \eqref{eq:expansion-beta} gives a well defined state also in the limit $\beta \to \infty$. } 
\begin{multline}\label{eq:expansion-beta}
\omega_{\beta,\la}^{U(i \frac{\beta}{2})}(A)  = \\
\sum_{{n_1, n_2=0}}^{+\infty} (-1)^{n_1+n_2} \int_{\frac{\beta}{2} S_{n_1}} dU \int_{\frac{\beta}{2} S_{n_2}} dV  \omega_{\beta,\la}^c(K_{iu_{n_1}-i\frac{\beta}{2}}\otimes \dots \otimes K_{iu_1-i\frac{\beta}{2}} \otimes A\otimes K_{iv_{n_2}}\otimes \dots \otimes K_{iv_1}).
\end{multline}
The following proposition gives the connection of the limit $\beta\to\infty$ of $\omega_{\beta,\la}^{U(i \frac{\beta}{2})}$ with the limit $t\to\infty$ 
of $\omega_{\la}^{U(-t)}$.

\begin{proposition}\label{prop:limtlimbeta}
Whenever $h$ in $V_{\chi,h}$ is of compact support, it holds that
\begin{equation}\label{eq:ret-equilibrium}
\lim_{t\to\infty} \omega_\la^{U(-t)}(R_{V_{\chi,h}}(A)) 
=\lim_{\beta\to\infty}\omega_{\la,\beta}^{U(i\frac{\beta}{2})}(R_{V_{\chi,h}}(A)),
\end{equation}
meaning that both limits exist and coincide.
\end{proposition}
\begin{proof}
We start by observing that by Prop.~\ref{pr:time-decay} in Appendix~\ref{sec:bounds}, there
holds that, when evaluated on functionals  in $\cA$, $\lim_{\beta\to\infty}  \omega_{\la,\beta} = \omega_\la$, which entails, for $t > 0$, 
\begin{equation}\label{eq:ground}
\lim_{\beta\to\infty} \omega_{\beta, \la}^{U(-t)}(A)
= \omega_{\la}^{U(-t)}(A), \qquad A\in\mathcal{A}.
\end{equation}
We now prove that we can take the limit for $t \to \infty$ of the previous equation and that the order in which the limits $t\to\infty$ and $\beta\to \infty$ of $ \omega_{\beta, \la}^{U(-t)}(A)$ are taken does not influence the result. In order to do this we show that the above limit is uniform in $t > 0$ by comparing the expansions~\eqref{eq:expansion-commutators} for $\omega_{\beta, \la}^{U(-t)}(A)$ and $\omega_{\la}^{U(-t)}(A)$ and estimating the difference term by term using again the bounds on the decay of the propagators in the time direction given in Prop.~\ref{pr:time-decay}. To begin with, we observe that, arguing as in the first part of the proof of Proposition~\ref{pr:clustering-alfav}, and using the notations introduced there, the $n$-th order of the expansion~\eqref{eq:expansion-commutators} relative to $\omega^{U(-t)}_{\beta,\la}(A)$ can be written as a finite linear combination of terms of the form
\begin{equation}\label{eq:expG}
\int_{t S_n} dT \,M  \bigg(\prod_{l \in E(G)} \Gamma_{s(l),r(l)}\bigg)\left(\tilde{K}_{-t_1} \otimes \tilde{K}_{-t_{2}} \dots \otimes \tilde{K}_{-t_n}\otimes \tilde{A}\right)
\end{equation}
where $G$ is a connected graph with $n+1$ vertices in correspondence with $\{\tilde K_{-t_1}, \dots, \tilde K_{-t_n}, \tilde A\}$ and, as discussed in the proof of Proposition~\ref{pr:clustering-alfav} 
$\tilde{A} := e^{\frac{1}{2} \int dxdy  (\Delta_{\beta,\la}(x-y) - \Delta_{+,\la}(x-y)) \frac{\delta^2}{\delta{\phi(x)}\delta{\phi(y)}} }A$.
A similar expansion holds of course also for $\omega^{U(-t)}_\la(A)$, with $\tilde K_{-t_j}$, $\tilde A$ and $\Gamma$ replaced by $K_{-t_j}$, $A$ and
\[
\Gamma_{\Delta_{+,\la}}=\int_{M^2} dxdy\, \Delta_{+,\la}(x-y) \frac{\delta}{\delta \phi(x)}\otimes \frac{\delta}{\delta\phi(y)}
\]
respectively. One verifies easily that the map $A \mapsto \tilde A$ deforms the $\star_\la$ and $\cdot_{T_\la}$ products as
\begin{align*}
(A\star_\la B)\tilde{\,} &= m e^{\Gamma} (\tilde A \otimes \tilde B),\\
(A \cdot_{T_\la} B)\tilde{\,} &= m e^{\Gamma_{\Delta_{F,\la}+\Delta_{\beta,\la}-\Delta_{+,\la}}}(\tilde A \otimes \tilde B).
\end{align*}
Moreover,
\begin{multline*}
e^{\frac{1}{2} \int dxdy  (\Delta_{\beta,\la}(x-y) - \Delta_{+,\la}(x-y)) \frac{\delta^2}{\delta{\phi(x)}\delta{\phi(y)}} } \phi(x)^n \\
= \sum_{k\leq \frac n 2} \frac {n(n-1)\dots(n-2k+1)}{2^k k!} \phi(x)^{n-2k} (\Delta_{\beta,\la}-\Delta_{+,\la})(0)^k.
\end{multline*}
Thanks to these observations and to Prop.~\ref{prop:bogoexpansion}, one sees that at perturbative order $\nu_j$, each factor $\tilde K_{-t_j}$ appearing in~\eqref{eq:expG} can be expanded in a finite linear combination of terms of the form
\begin{equation}\label{eq:Kexp}
 m_{\beta,\la}^{\alpha_j}\int_{M^{\nu_j}} dX^{(j)}\, (\dot\chi^- h)(x^{(j)}_1) \prod_{k=2}^{\nu_j} (\chi h)(x^{(j)}_k) D_{\beta,j}(X^{(j)}) \Phi_{t_j}(X^{(j)}),
\end{equation}
with $m_{\beta,\la} := (\Delta_{\beta,\la}-\Delta_{+,\la})(0)$, $\alpha_j \in \bN$, $X^{(j)} = (x^{(j)}_1,\dots,x^{(j)}_{\nu_j}) \in M^{\nu_j}$, $D_{\beta,j}$ a product of factors of the form $(\Delta_{F,\la}+\Delta_{\beta,\la}-\Delta_{+,\la})(x^{(j)}_k-x^{(j)}_i)$ or $\Delta_{\beta,\la}(x^{(j)}_k-x^{(j)}_i)$, and $\Phi_{t_j}$ a monomial in the fields $\phi(x^{(j)}_1-t_j e_0), \dots, \phi(x^{(j)}_{\nu_j}-t_j e_0)$. Inserting these expansions into~\eqref{eq:expansion-commutators} and assuming, without loss of generality, that $A$ is as in Eq.~\eqref{eq:A},
\[
A(\phi) = \int_{M^\ell} dy_1\dots dy_\ell f(y_1,\dots ,y_\ell)\phi^{p_1}(y_1)\dots \phi^{p_\ell}(y_\ell), 
\]
with compact $\supp f$, we see that at perturbative order $\nu = \nu_1 + \dots + \nu_n$ the difference $\omega_{\beta,\la}^{U(-t)}(A) - \omega_{\la}^{U(-t)}(A)$ can be written as a finite linear combination of terms of the form
\begin{multline}\label{eq:expansion-differ1}
\int_{tS_n} dT \int_{M^{\nu+\ell}}dX dY\, \prod_{j=1}^n\left\{(\dot\chi^- h)(x^{(j)}_1) \prod_{k=2}^{\nu_j} (\chi h)(x^{(j)}_k)m_{\beta,\la}^{\alpha_j}D_{\beta,j}(X^{(j)}) \right\} f(Y) \times \\
\left[\prod_{l \in E(G)} \Delta_{\beta,\la}\big(z_{\rho(l)}-z_{\sigma(l)}-(t_{r(l)}-t_{s(l)})e_0\big)\right],
\end{multline}
with $\alpha_j > 0$, or of the from
\begin{multline}\label{eq:expansion-differ2}
\int_{tS_n} dT \int_{M^{\nu+\ell}}dX dY\, \prod_{j=1}^n\left\{(\dot\chi^- h)(x^{(j)}_1) \prod_{k=2}^{\nu_j} (\chi h)(x^{(j)}_k) \right\} f(Y) \times \\
\left[\prod_{j=1}^n D_{\beta,j}(X^{(j)})\prod_{l \in E(G)} \Delta_{\beta,\la}\big(z_{\rho(l)}-z_{\sigma(l)}-(t_{r(l)}-t_{s(l)})e_0\big)\right.\\
\left.-\prod_{j=1}^n D_{j}(X^{(j)})\prod_{l \in E(G)} \Delta_{+,\la}\big(z_{\rho(l)}-z_{\sigma(l)}-(t_{r(l)}-t_{s(l)})e_0\big)\right],
\end{multline}
where $X = (X^{(1)},\dots, X^{(n)}) \in M^\nu$, $Y = (y_1,\dots,y_\ell) \in M^\ell$, $Z = (z_1,\dots,z_{\nu+\ell}) = (X,Y) \in M^{\nu+\ell}$, $\sigma, \rho : E(G) \to \{1,\dots,\nu+\ell\}$ are such that $\sigma(l) \in \{\nu_1+\dots +\nu_{s(l)-1}+1,\dots,\nu_1+\dots+\nu_{s(l)}\}$,  $\rho(l) \in \{\nu_1+\dots +\nu_{r(l)-1}+1,\dots,\nu_1+\dots+\nu_{r(l)}\}$ with $\nu_0 := 0$, $\nu_{n+1} := \ell$, where $t_{n+1} := 0$, and where $D_j$ is obtained from $D_{\beta,j}$ by replacing $\Delta_{F,\la}+\Delta_{\beta,\la}-\Delta_{+,\la}$ with $\Delta_{F,\la}$ and $\Delta_{\beta,\la}$ with $\Delta_{+,\la}$. We now show that the expression~\eqref{eq:expansion-differ2} vanishes as $\beta \to +\infty$, uniformly in $t \in \bR$. To this end, we notice that, given $a > 2\epsilon$ and a compact set $\Lambda \subset \Sigma_0$ such that $\supp h \subset \Lambda$, $\supp f \subset [-a,a]\times \Lambda$,
by the estimates in Prop.~\ref{pr:time-decay}, the considered expression can be bounded by
\begin{multline*}
 C^{|E(G)|}\|f\|_\infty\int_{([-a,a]\times \Lambda)^{\nu+\ell}} dZ \left\{\left|\prod_{j=1}^n D_{\beta,j}(X^{(j)})-\prod_{j=1}^n D_{j}(X^{(j)})\right|+e^{-\beta m}\prod_{j=1}^n \left|D_{j}(X^{(j)})\right| \right\}\times\\
 \int_{tS_n} dT \prod_{l \in E(G)} \frac{1}{1+|z_{\rho(l)}^0-z_{\sigma(l)}^0-(t_{r(l)}-t_{s(l)})|^{3/2}}.
\end{multline*}
for some constant $C > 0$ depending only on $\Lambda$.
Now by an argument similar to the one in the second part of the proof of Proposition~\ref{pr:clustering-alfav}  the $T$ integral in the above expression can be bounded uniformly in $Z^0 = (z_1^0,\dots z_{\nu+\ell}^0) \in \bR^{\nu+\ell}$ and in $t > 0$ by a constant. Therefore, 
by the boundedness of $\Delta_{+,\la}$, $\Delta_{F,\la}$ (Prop.~\ref{pr:prop} and \ref{pr:prop1}), the above expression is estimated by a $(t,\beta)-$independent constant times
\[
\int_{([-a,a]\times\Lambda)^\nu}dX\,\left|\prod_{j=1}^n D_{\beta,j}(X^{(j)})-\prod_{j=1}^n D_{j}(X^{(j)})\right|+ e^{-\beta m}(2a|\Lambda|)^\nu\prod_{j=1}^n \|D_j\|_\infty\,, 
\]
which, taking into account that $D_{\beta,j}(X^{(j)}) \to D_j(X^{(j)})$ again by Prop.~\ref{pr:time-decay}, shows the stated vanishing of~\eqref{eq:expansion-differ2} as $\beta \to +\infty$ uniformly in $t$. The expression~\eqref{eq:expansion-differ1} can be treated similarly, exploiting the fact that $m_{\beta,\la}^{\alpha_j} \to 0$ for $\alpha_j > 0$. This finally shows that Eq.~\eqref{eq:ground} holds
uniformly in time. We conclude that 
\[
\lim_{t\to\infty} \lim_{\beta\to \infty }  \omega_{\beta, \la}^{U(-t)}(A) = \lim_{\beta\to \infty } \lim_{t\to\infty}  \omega_{\beta, \la}^{U(-t)}(A).
\]
Moreover, as shown in Proposition \ref{pr:return-to-equilibrium} return to equilibrium holds, hence
\[
\lim_{t \to \infty} \omega_{\beta,\la}^{U(-t)}(A)=\lim_{t\to\infty} \omega_{\beta,\la}(\alpha_t^V(A)) = \omega_{\beta,\la}^{U(i\frac{\beta}{2})}(A), \qquad A\in \mathcal{A}.
\]  
Hence, combining all these observations we have that the two limits in the statement exist and
\[\begin{split}
\lim_{t\to\infty} \omega_{\la}(\alpha_t^V(A)) &= \lim_{t\to\infty}  \omega_{\la}^{U(-t)}(A) =
\lim_{t\to\infty} \lim_{\beta\to \infty }  \omega_{\beta, \la}^{U(-t)}(A) \\
&= 
\lim_{\beta\to \infty } \lim_{t\to\infty}  \omega_{\beta, \la}^{U(-t)}(A) =
\lim_{\beta\to \infty }   \omega_{\beta, \la}^{U(i\frac{\beta}{2})}(A) 
\end{split}\]
and we have the thesis.
\end{proof}

We recall that the right hand side of \eqref{eq:ret-equilibrium} can now be expanded as in \eqref{eq:expansion-beta}. As discussed in the proof of the previous proposition, at finite perturbative order, formula~\eqref{eq:expansion-beta} is well posed thanks to the analyticity properties of $\omega_{\la,\beta}^c$. Furthermore, the limit $\beta\to\infty$ can be taken in view of the exponential decaying for imaginary times of the correlation function discussed in Proposition \ref{pr:clustering}.
Notice that equation \eqref{eq:ret-equilibrium} expresses the return to equilibrium property for the equilibrium state at vanishing temperature, namely for the vacuum state. The return to equilibrium property for the equilibrium states has been discussed in \cite{DFP}, see also Proposition \ref{pr:return-to-equilibrium} for a proof in the case of noncommutative spacetime. 
Now we observe that the limit $h\to1$ of the right hand side of \eqref{eq:expansion-beta} can be taken.
It is furthermore important to stress that the limit $h\to 1$ of \eqref{eq:expansion-beta} gives a state which is $\chi$ independent.

We have actually the following theorem which establishes the existence of a ground state in 
 the adiabatic limit 
\begin{theorem}
Consider the functional $\tilde{\omega}_\la$ on the algebra $(\cA(O), \star_{\la,V})$ of interacting fields supported in $\mathcal{O}\subset \Sigma_{0,\epsilon}$ defined as
\begin{equation}\label{eq:omegatilde}
\tilde{\omega}_\la(A):= \lim_{h\to 1, \beta\to \infty}\omega_{\la,\beta}^{U(i\frac{\beta}{2})}(R_{V_{\chi,h}}(A)).
\end{equation}
We have that $\tilde{\omega}_\la$ is finite at all perturbative orders, it does not depend on the order in which the limits are taken,
it is positive and normalized in the sense of perturbation theory and it is invariant under interacting spacetime translations.
\end{theorem}
\begin{proof}
The functional $\tilde{\omega}_\la$ is a limit of states given in the sense of perturbation theory, hence, if the limit exists it must be a state in the sense of perturbation theory. In particular it must be positive and normalized. 
Consider now the expansion given in~\eqref{eq:expansion-beta}
\begin{multline}\label{eq:omega-beta-espansione}
\omega_{\la,\beta}^{U(i \frac{\beta}{2})}(R_{V_{\chi,h}}(A))  \\
= \sum_{{n_1,n_2}}  \int_{\frac{\beta}{2} S_{n_1}}\!\!\!\! dU \int_{\frac{\beta}{2} S_{n_2}}\!\!\!\! dV  \omega_{\la,\beta}^c(K_{iu_{n_1}-i\frac{\beta}{2}}\otimes \dots \otimes K_{iu_1-i\frac{\beta}{2}} \otimes R_{V_{\chi,h}}(A)\otimes K_{iv_{n_2}}\otimes \dots \otimes K_{iv_1}).
\end{multline}
Notice that the lowest perturbation order in $K$ is one, hence, at a given perturbative order the sum over ${n_1,n_2}$ contains only finitely many terms and only finitely many terms of the expansion of $K$ and $R_{V_{\chi,h}}(A)$ can contribute. 
Finally we observe that all these terms can be expanded in a sum of connected graphs as explained in the proof of Proposition~5 in~\cite{FredenhagenLindner}. 
In view of the linearity of the state and of $R_{V_{\chi,h}}$ to analyze the generic contribution in this expansion in graphs it is sufficient to consider  
\[
A(\phi) = \int_{M^\ell} dy_1\dots dy_\ell f(y_1,\dots ,y_\ell)\phi^{p_1}(y_1)\dots \phi^{p_\ell}(y_\ell),
\]
where $f$ is a compactly supported smooth function. 
At order $\nu$ in perturbation theory, with $n=n_1+n_2 \leq \nu$, a generic contribution to the expansion in a sum of connected graphs of \eqref{eq:omega-beta-espansione} can be written, similarly to~\eqref{eq:expansion-differ1}, as
\begin{multline*}
F=\int_{\frac \beta 2S_{n_1}} dU \int_{\frac \beta 2 S_{n_2}} dV\int_{M^{\nu+\ell}}dX dY\, \prod_{j=1}^n\left\{(\dot\chi^- h)(x^{(j)}_1) \prod_{k=2}^{\nu_j} (\chi h)(x^{(j)}_k) m_{\beta,\la}^{\alpha_j}D_{\beta,j}(X^{(j)})\right\} f(Y) \times \\
\prod_{l \in E(G)} \Delta_{\beta,\la}\big(z_{\rho(l)}-z_{\sigma(l)}+i(w_{r(l)}-w_{s(l)})e_0\big)
\end{multline*}
where now $U=(u_1,\dots, u_{n_1})$, $V=(v_1,\dots, v_{n_2})$, 
the domain of $(U, V)$ integration is such that
\[
0 <u_{n_1}<\dots  <u_1< \frac{\beta}{2}, \qquad 0 <v_{n_2}<\dots  <v_1< \frac{\beta}{2},
\] 
and $w_j = u_{n_1-j+1}-\frac \beta 2$, $j=1,\dots, n_1$, $w_{n_1+1} = 0$, $w_j = v_{n_1+n_2-j+2}$, $j=n_1+2,\dots,n_1+n_2+1$.
Then, with $\Lambda \subset \Sigma_0$ and $a > 0$ as in the proof of the previous proposition, using the bound of $\Delta_{\beta, \lambda}(t,\mathbf{x})$ given in \eqref{eq:thermal-propagator} of Proposition \ref{pr:clustering}, we obtain for $h = 1$ the estimate
\begin{multline*}
|F| \leq \|f\|_\infty\int_{[-a,a]^{\nu+\ell}}dZ^0\int_{\Sigma_0^\nu}d\mathbf{X} \int_\Lambda d\mathbf{Y}\, \prod_{j=1}^n m_{\beta,\la}^{\alpha_j}|D_{\beta,j}(X^{(j)}) |
\prod_{l \in E(G)} e^{-\frac m 2 |\mathbf{z}_{\rho(l)}-\mathbf{z}_{\sigma(l)}|} \times \\
\int_{\frac \beta 2 S_{n_1}}dU\int_{\frac \beta 2 S_{n_2}} dV\prod_{l \in E(G)}J(w_{r(l)}-w_{s(l)}),
\end{multline*} 
where
\[
J(\mathrm{Im}\, t):=\left(\frac{e^{-\frac{m}{2}|\mathrm{Im}\, t|}}{1-e^{-\beta m}}+\frac{e^{-\frac{m}{2}(\beta-|\mathrm{Im}\, t|)}}{1-e^{-\beta m}}\right),
\]
and where, in view of the definition of connected correlation functions,  the orientation of the edges $l$ of $G$ has to be chosen in such a way that $r(l) < s(l)$, and then, in view of the form of the domain of integration, only arguments $w_{r(l)}-w_{s(l)} < 0$ appear in the $J$'s product.
The graph $G$ has $n+1$ vertices corresponding to $R_V(A), K_{i(u_j-\frac{\beta}{2})}, K_{iv_j}$ 
of \eqref{eq:omega-beta-espansione} and is connected, and thus we can find a sub graph $\gamma$ which is a rooted tree with root in the vertex corresponding to $R_V(A)$.     
    All the $J$s corresponding to edges of $G$ not contained in $\gamma$ can be bound by a constant. In view of the form of $J$, we can now prune the leaves of this rooted tree taking the integral over the corresponding $w_j$ from $w_{j-1}$ to $w_{{j+1}}$ for $j=2,\dots,n_1+n_2$, $j \neq n_1+1$, or  over $w_1$ from $-\frac \beta 2$ to $w_2$, or over $w_{n_1+n_2+1}$ from $w_{n_1+n_2}$ to $\frac \beta 2$.   
    This integral is finite and bounded uniformly in $\beta$: 
   in particular, for the leaf $w_i$ joined with the vertex $w_j$ with $w_i>w_j$ the domain of integration is a subset of the interval $[w_j,w_j+\beta]$ and we have thus
\[
\int dw_i\, J(w_i-w_j)  \leq \int_0^\beta ds \left(\frac{e^{-\frac{m}{2}s}}{1-e^{-\beta m}}+\frac{e^{-\frac{m}{2}(\beta-s)}}{1-e^{-\beta m}}\right) = 
\frac 4 m \frac{1-e^{-\frac{\beta m}{2}}}{1-e^{-\beta m}}=\frac{4}{m} \frac{1}{1+e^{\frac{\beta m}{2}}}\leq \frac 4 m,
\]    
while if the leaf $w_l$ is joined with the vertex $w_k$ with $w_k>w_l$ the domain of $w_l$-integration is a subset of the interval $[-\beta+w_k ,w_k]$ and we get the same estimate for the corresponding integral. We can repeat the procedure till when all the $w_j$ integrations are taken. The final result is bounded by a constant $E > 0$ uniformly in $\beta$ and thus we get for $F$ the estimate
\begin{equation}\label{eq:Fbound}
|F| \leq E \|f\|_\infty\int_{[-a,a]^{\nu+\ell}}dZ^0\int_{\Sigma_0^\nu}d\mathbf{X} \int_\Lambda d\mathbf{Y}\, \prod_{j=1}^n m_{\beta,\la}^{\alpha_j}|D_{\beta,j}(X^{(j)}) |
\prod_{l \in E(G)} e^{-\frac m 2 |\mathbf{z}_{\rho(l)}-\mathbf{z}_{\sigma(l)}|}.
\end{equation}
Now Prop.~\ref{prop:bogoexpansion}, taking into account the fact that $\dot V$ is local, shows that the all the graphs contributing to the expansion of $K = R_V(\dot V)$ are connected. Therefore the graph $\tilde G$ obtained by substituting the vertices of $G$ with the graphs corresponding to the functions $D_{\beta,j}$, $j=1,\dots,n$, is also connected. Moreover by Prop.~\ref{pr:clustering} the factors $(\Delta_{F,\la}+\Delta_{\beta,\la}-\Delta_{+,\la})(x^{(j)}_k-x^{(j)}_i)$ and $\Delta_{\beta,\la}(x^{(j)}_k-x^{(j)}_i)$ 
appearing in $D_{\beta,j}(X^{(j)})$ can be bounded by a constant times $e^{-\frac m 2 |\mathbf{x}^{(j)}_k-\mathbf{x}^{(j)}_i|}$ uniformly for the time components in $[-a,a]$ and $\beta > 0$. Then the integral in~\eqref{eq:Fbound} can be in turn bounded by an integral which can be shown to be finite following the proof of Theorem \ref{thm:adiabatic} and its Corollary \ref{cor:adiabatic}. Summing up, the above estimates show that, order by order in the perturbative expansion, the limit $\lim_{h \to 1}\omega_{\la,\beta}^{U(i\frac{\beta}{2})}(R_{V_{\chi,h}}(A))$ exists and is uniform in $\beta >0$. And since $\lim_{\beta \to +\infty}\omega_{\la,\beta}^{U(i\frac{\beta}{2})}(R_{V_{\chi,h}}(A))$ exists for compactly supported $h$ by Proposition~\ref{prop:limtlimbeta}, we conclude that the limit~\eqref{eq:omegatilde} exists, and the order of the limits $h\to 1$ and $\beta \to +\infty$ is immaterial.
 
Finally, the obtained state is invariant under spacetime translation because it is the limit for $h\to1$ of states which are invariant under time translation thanks to \eqref{eq:ret-equilibrium} and can be interpreted as the ground state of the interacting theory.
\end{proof} 

We remark that the limit $h \to 1$ was also possible in the approach of~\cite{BDFP03}, but there it was not clear how to control the limit $\chi \to 1$. The fact that the latter limit is not necessary in the present framework is a remarkable feature of the perturbative AQFT machinery, whose application to QFT on QST was made possible thanks to the redefinition of the effective interaction in Sec.~\ref{sec:effective}.

\section{Outlook}
The adiabatic limit of the perturbation expansion of vacuum expectation values of interacting fields (known as weak adiabatic limit) that we have proved to exist in the case of a self-interacting scalar field provides us with an ultraviolet finite theory. It is conceivable that a similar result could also be obtained for physically more interesting theories, like QED. 

However since the resulting theory is nonlocal, the standard results of scattering theory do not directly apply, and the adiabatic limit of the $S$ matrix elements (the so called strong adiabatic limit) remain to be investigated. Moreover in order to compare the results of this procedure with observable data in physically interesting situations a finite renormalization would be needed. For the potentially divergent terms when $\la \to 0$ would give large unphysical contributions.

A possible way to fix the values of the renormalization constants could be the following.
The renormalization constants should be calculated as a function of the Planck length in such a way to cancel those contributions of the irreducible diagrams which would diverge in the limit where  the Planck length vanishes; in that limit the renormalized perturbation series of the usual theory ought to be recovered.  In particular gauge invariance and Lorentz covariance, both broken by our prescription for interaction, would be recovered in this limit. This and related problems will be studied elsewhere. 

It is a deep open problem whether it is possible to define interactions of quantum fields on QST in such a way as to preserve Lorentz and gauge invariance. 

It must be anyway remarked that the corrections to the usual theory that would be obtained by the methods described here for the physical value of the Planck length are expected to be dramatically small, as it has been seen to be the case even when they might cumulate in astrophysical situations \cite{DMP2}.

\appendix

\section{Basic properties of modified propagators}\label{se:prop}

We discuss in this appendix some elementary properties enjoyed by the modified propagators used in this paper. 
We recall that
\begin{align}
\hat{\Delta}_{\lambda}(p) &= \hat{\Delta}(p) e^{-{\lambda^2}\langle p\rangle^2} = -\frac{i}{(2\pi)^3}  \delta(p^2+m^2) \varepsilon(p_0) e^{-{\lambda^2}\langle p\rangle^2},
\\ \label{eq:delta+}
\hat{\Delta}_{+,\lambda}(p) &=  \hat{\Delta}_+(p) e^{-{\lambda^2}\langle p\rangle^2} = \frac{1}{(2\pi)^3}  \delta(p^2+m^2) \theta(p_0) e^{-{\lambda^2}\langle p\rangle^2},
\end{align}
which entails
\begin{align}\label{eq:delta}
\Delta_{\la}(x) &= \frac{1}{(2\pi)^3} \int_{\bR^3} \frac{d\mathbf{p}}{\omega(\mathbf{p})} e^{-\la^2(2|\mathbf{p}|^2+m^2)}\sin(\omega(\mathbf{p})t-\mathbf{p}\cdot\mathbf{x}), \\
\Delta_{+,\la}(x) &= \frac{1}{(2\pi)^3} \int_{\bR^3} \frac{d\mathbf{p}}{2\omega(\mathbf{p})} e^{-\la^2(2|\mathbf{p}|^2+m^2)}e^{-i(\omega(\mathbf{p})t-\mathbf{p}\cdot\mathbf{x})},
\end{align}
where, as customary, $\omega(\mathbf{p})=\sqrt{|\mathbf{p}|^2+m^2}$.
\begin{proposition}\label{pr:prop}
The distributions $\Delta_\lambda$ and $\Delta_{+,\lambda}$ given in \eqref{eq:deltalambda} and in \eqref{eq:delta+lambda} have the following properties:
\begin{itemize}
\item[a)] $\Delta_\lambda$ and $\Delta_{+,\lambda}$ are smooth functions which are also in $L^\infty$.
\item[b)] $\Delta_\lambda$ and $\Delta_{+,\lambda}$ are solutions of the equation of motion \eqref{eq:free-equation}.
\end{itemize}
\end{proposition}
\begin{proof}
From the definition we have that
\[
\Delta_{\lambda}(x) = \langle\Delta(G_\lambda),G_{\lambda,x}\rangle = \Delta(G_\la * G_\la) (x)= \Delta(G_{\sqrt{2}\lambda})(x),
\]
where $G_{\lambda,x}(y)=G_\lambda(y-x)$ and $*$ denotes convolution.  
Since $\Delta$ is a map from Schwartz functions to smooth functions, we have that $\Delta(G_{\sqrt{2}\lambda})$ is a smooth function.
From Eq.~\eqref{eq:delta} we obtain
\[
\|\Delta_\lambda\|_\infty 
\leq
\frac{1}{(2\pi)^3} \int_{\bR^3} d\mathbf{p}
\frac{1}{\omega(\mathbf{p})}  e^{-\lambda^2 (2|\mathbf{p}|^2+m^2)}
\]
and the right hand side of the previous inequality is finite, so that $\Delta_\la$ is bounded. This proves $a)$ for $\Delta_\lambda$.
To prove $b)$ we notice that
\[
(\Box-m^2)\Delta_\lambda(x) = \langle\Delta((\Box-m^2)G_\lambda),G_{\lambda,x}\rangle.
\]
The thesis follows from the fact that  $\Delta$ is a weak solution of the equation of motion. 
The same proofs can be applied with minor modifications to study $\Delta_{+,\lambda}$.
\end{proof}

\begin{proposition}\label{prop:Feynman}
There holds, for the Fourier transform of the modified Feynman propagator, 
\[
\hat \Delta_{F,\la}(p) = -\frac{i}{(2\pi)^4}  \frac{1}{ p^2+m^2  -i\epsilon}
 e^{-\lambda^2 (2|\mathbf{p}|^2+m^2)}.
\]
\end{proposition}
\begin{proof}
Computing the inverse Fourier transform along $p_0$ of Eq.~\eqref{eq:delta+} one gets
\[
{{\Delta}}_{+,\lambda}(t,\mathbf{p}) 
= \frac{1}{(2\pi)^{7/2}} \frac{e^{-i\omega t}e^{-\lambda^2(2\omega^2-m^2)}}{2\omega}.
\]
The corresponding partial Fourier transform of the modified Feynman propagator is 
\[
{{\Delta}}_{F,\lambda}(t,\mathbf{p}) = \theta(t) {{\Delta}}_{+,\lambda}(t,\mathbf{p})  + \theta(-t){{\Delta}}_{+,\lambda}(-t,-\mathbf{p}) .
\]
Inserting an ultraviolet regulator, its Fourier transform along $t$ gives
\begin{align*}
\hat{{\Delta}}_{F,\lambda}(p_0,\mathbf{p}) 
&= 
\frac{1}{\sqrt{2\pi}}\int_0^\infty dt\; e^{-\epsilon t} e^{ip_0 t}\hat{{\Delta}}_{+,\lambda}(t,\mathbf{p})  + 
\frac{1}{\sqrt{2\pi}}\int_0^\infty dt\; e^{-\epsilon t} e^{-ip_0 t}\hat{{\Delta}}_{+,\lambda}(t,-\mathbf{p})
\\
&= 
\frac{1}{(2\pi)^4}\frac{e^{-\lambda (2\omega(\mathbf{p})^2-m^2)}}{2\omega(\mathbf{p}) } \int_0^\infty dt\; \left(
 e^{-\epsilon t +ip_0 t -i\omega(\mathbf{p}) t }  + 
 e^{-\epsilon t -ip_0 t -i\omega(\mathbf{p}) t } \right)
 \\
&= 
\frac{1}{(2\pi)^4}
\frac{e^{-\lambda^2 (2\omega(\mathbf{p})^2-m^2)}}{2\omega(\mathbf{p}) }
 \frac{2(\epsilon +i \omega(\mathbf{p}))}{(\epsilon +i \omega(\mathbf{p}))^2 +p_0^2}
 \\
&= 
\frac{i}{(2\pi)^4}
e^{-\lambda^2 (2\omega(\mathbf{p})^2-m^2)}
 \frac{1}{\epsilon^2 +2i\epsilon\omega(\mathbf{p}) -\omega(\mathbf{p})^2 +p_0^2}
\\
&= 
\frac{-i}{(2\pi)^4}
 \frac{1}{ p^2+m^2  -i\epsilon}
 e^{-\lambda^2 (2|\mathbf{p}|^2+m^2)},
\end{align*}
as requested.
\end{proof}

\section{Bounds for the modified propagators}\label{sec:bounds}

In this appendix, we discuss some bounds valid for the propagators used in this paper. The first proposition is about the decay in space while the second is about the decay in time.

\begin{proposition}\label{pr:clustering}
For each fixed $a >0$ the propagators $\Delta_{F,\lambda}$ and $\Delta_{+,\lambda}$ satisfy the estimates
\[
|D(t,\mathbf{x})|\leq C e^{-m |\mathbf{x}|} , \qquad D\in \{\Delta_{F,\lambda},\Delta_{+,\lambda}\} , \qquad t\in (-a,a)
\]
where $C > 0$ is a constant (depending on $a$ and $\la$).
Furthermore, the function $t\mapsto \Delta_{+,\lambda}(t,\mathbf{x})$ is an entire analytic function. For $\emph{Im}\,t<0$, it holds that 
\[
|\Delta_{+,\lambda}(t,\mathbf{x})|\leq C e^{-m r} , \qquad \emph{Re}\,t\in (-a,a), \qquad r = \frac{1}{2}\left({|\text{\emph{Im}}\,t|+|\mathbf{x}|}\right)
\]
where $C>0$ is the same constant as above. Similarly, the function $t\mapsto \Delta_{\beta,\lambda}(t,\mathbf{x})$ constructed with the thermal two-point function $\Delta_{\beta,\la}$, is an entire analytic function and for $\text{Im}\,t<0$ it holds that 
\begin{equation}\label{eq:thermal-propagator}
|\Delta_{\beta,\lambda}(t,\mathbf{x})|\leq C e^{-\frac{m}{2}|\mathbf{x}|}\left(\frac{e^{-\frac{m}{2}|\emph{Im}\,t|}}{1-e^{-\beta m}}+\frac{e^{-\frac{m}{2}(\beta-|\emph{Im}\,t|)}}{1-e^{-\beta m}}\right) , 
\qquad \emph{Re}\,t\in (-a,a),
\end{equation}
where $C>0$ is the same constant as above.
\end{proposition}
\begin{proof}
We discuss the decay properties of $\Delta_{+,\lambda}$. The definition \eqref{eq:feynman} implies that the same result will then hold for $\Delta_{F,\lambda}$.
In view of the rotation invariance of $\Delta_{+,\lambda}(t,\mathbf{x})$ we just need to analyze the decay in the direction $\mathbf{x} = (x^1,0,0)$ for large $x^1$. 
To this end, we recall that
\begin{equation}\label{eq:delta+p}
\Delta_{+,\lambda}(t,\mathbf{x})
= 
\frac{e^{-\lambda^2 m^2}}{(2\pi)^3} \int_{\mathbb{R}^3}  \frac{ e^{-2\lambda^2 |\mathbf{p}|^2}e^{-it \sqrt{|\mathbf{p}|^2+m^2}} e^{i\mathbf{p}\mathbf{x}}}{2\sqrt{|\mathbf{p}|^2+m^2}}   d\mathbf{p}.
\end{equation}
We observe that the integrand, seen as a function of $p_1$ is analytic in the strip $\text{Im}(p_1)\in (-m,m)$: actually the cut of $\sqrt{|\mathbf{p}|^2+m^2}$ as well as the poles of $1/\sqrt{|\mathbf{p}|^2+m^2}$ are contained in the region  $|\text{Im}\,p_1|\geq \sqrt{m^2+|p_2|^2+|p_3|^2}$. Furthermore, the integrand vanishes for large $p_1$ in that strip. Hence we consider, for $\epsilon\in (0,m)$, with $q:=m-\epsilon$,
\begin{align*}
e^{x^1 (m-\epsilon)}\Delta_{+,\lambda}(t,\mathbf{x})
&= 
\frac{e^{-\lambda^2 m^2}}{(2\pi)^3} \int_{\mathbb{R}^3}  \frac{ e^{-2\lambda^2 |\mathbf{p}|^2}e^{-it \sqrt{|\mathbf{p}|^2+m^2}} e^{i(\mathbf{p}\mathbf{x}-i(m-\epsilon) x^1)}}{2\sqrt{|\mathbf{p}|^2+m^2}}   d\mathbf{p}\\
&= 
\frac{e^{-\lambda^2 m^2}}{(2\pi)^3} \int_{\mathbb{R}^3}  \frac{ e^{-2\lambda^2 (|\mathbf{p}|^2+2iqp_1-q^2)}e^{-it \sqrt{|\mathbf{p}|^2+2iqp_1-q^2+m^2}} e^{i\mathbf{p}\mathbf{x}}}{2\sqrt{|\mathbf{p}|^2+2iqp_1-q^2+m^2}}   d\mathbf{p}
\end{align*}
where in the last equality we have used the residue theorem to move the integration in $p_1$ from the real line to the line $\mathrm{Im}\, p_1 = q$.
We obtain then the following bound
\begin{align}
\left|e^{x^1 (m-\epsilon)}\Delta_{+,\lambda}(t,\mathbf{x})\right|
&\leq 
\frac{e^{\lambda^2 (2q^2-m^2)}}{(2\pi)^3} \int_{\mathbb{R}^3}  \frac{ e^{-2\lambda^2 |\mathbf{p}|^2}e^{|t|\left[(|\mathbf{p}|^2+m^2-q^2)^2+4q^2p_1^2\right]^{\frac 1 4} }}{\sqrt{p_2^2+p_3^2}}   d\mathbf{p}
\notag
\\
\label{eq:estimate-minus}
&\leq
\frac{e^{\lambda^2 m^2}}{(2\pi)^3}\int_{\mathbb{R}^3}  \frac{ e^{-2\lambda^2 |\mathbf{p}|^2}e^{a\left[(|\mathbf{p}|^2+m^2)^2+4m^2p_1^2\right]^{\frac 1 4} }}{\sqrt{p_2^2+p_3^2}}   d\mathbf{p} =: C,
\end{align}
where the constant $C$ does not depend on $\epsilon$ nor on $t$ for $t\in (-a,a)$. By the same argument we can also estimate $e^{-x^1(m-\epsilon)}\Delta_+(t,\mathbf{x})$ thus concluding the first part of the proof.

The analyticity of the function $t\mapsto \Delta_{+,\lambda}(t,\mathbf{x})$ is manifest in \eqref{eq:delta+p}. To prove the second estimate stated in the proposition we observe that for $q=\frac{m}{2}$, 
\[
0\leq\sqrt{|\mathbf{p}|^2-q^2+m^2}-\frac{m}{2} \leq \text{Re}\sqrt{|\mathbf{p}|^2+2iqp_1-q^2+m^2} -\frac{m}{2}
\]
hence, for positive large $x^1$, positive large $\beta$ and real $t$, adapting the estimate \eqref{eq:estimate-minus} we get
\begin{equation*}
\left|e^{(\beta+x^1) \frac{m}{2}}\Delta_{+,\lambda}(t-i\beta,\mathbf{x})\right|
\leq
\frac{e^{\lambda^2 m^2}}{(2\pi)^3}\int_{\mathbb{R}^3}  \frac{ e^{-2\lambda^2 |\mathbf{p}|^2}e^{a\left[(|\mathbf{p}|^2+m^2)^2+4m^2p_1^2\right]^{\frac 1 4} }}{\sqrt{p_2^2+p_3^2}}  
 d\mathbf{p} =: C
\end{equation*}
with the same constant as before. 

The analyticity of the function $t\mapsto \Delta_{\beta,\lambda}(t,\mathbf{x})$ holds because of the presence of the $e^{-2\lambda^2 |\mathbf{p}|^2}$ in its Fourier expansion. 
The last estimate can be obtained in a similar way starting from \eqref{eq:two-point-KMS} 
\begin{multline*}
\Delta_{\beta,\lambda}(t,\mathbf{x}) = 
\frac{e^{-\lambda^2 m^2}}{(2\pi)^3} \int_{\mathbb{R}^3}  \frac{ e^{-2\lambda^2 |\mathbf{p}|^2} e^{i\mathbf{p}\mathbf{x}}}{2\sqrt{|\mathbf{p}|^2+m^2}}
\left(e^{-it \sqrt{|\mathbf{p}|^2+m^2}}b_+(\sqrt{|\mathbf{p}|^2+m^2}) \right.\\
\left.+\,e^{it \sqrt{|\mathbf{p}|^2+m^2}}b_-(\sqrt{|\mathbf{p}|^2+m^2}) \right)   d\mathbf{p}
\end{multline*}
and bounding the Bose factors $b_\pm$ as follows
\[
b_+(\sqrt{|\mathbf{p}|^2+m^2}) =\frac{1}{1-e^{-\beta\sqrt{|\mathbf{p}|^2+m^2}}}   \leq \frac{1}{1-e^{-\beta m}},
\]
\[
b_-(\sqrt{|\mathbf{p}|^2+m^2}) =\frac{1}{e^{\beta\sqrt{|\mathbf{p}|^2+m^2}}-1}   \leq \frac{e^{-\beta {m}}}{1-e^{-\beta m}}\leq \frac{e^{-\beta \frac{m}{2}}}{1-e^{-\beta m}}.
\]
\end{proof}

\begin{proposition}\label{pr:time-decay}
Consider a compact set $\Lambda \subset \Sigma_0$. 
The propagators $\Delta_{+,\lambda}$, $\Delta_{\beta,\lambda}$ satisfy the estimates
\[
|D(t-iu,\mathbf{x})|\leq \frac{C}{1+|t|^{\frac{3}{2}}}  , \qquad D\in \{\Delta_{+,\lambda},\Delta_{\beta,\lambda}\}, 
\]
valid for $t \in \bR$, $u\in[0,\beta]$ and $\mathbf{x}\in \Lambda$, 
where $C > 0$ is a constant which depends on $\la$  and $\Lambda$ but not on $\beta$ for $\beta > \frac{1}{m}$.
Furthermore, 
\[
|\Delta_{\beta,\lambda}(t,\mathbf{x})-\Delta_{+,\lambda}(t,\mathbf{x})|\leq \frac{C}{1+|t|^{\frac{3}{2}}}e^{-\beta m}, \quad (t,\mathbf{x}) \in \bR\times\Lambda,
\]
where $C$ is the same constant as before.
\end{proposition}
\begin{proof}
Consider
\[
\Delta_{+,\la}(t-iu, \mathbf{x}) := \frac{e^{-\lambda^2 m^2}}{(2\pi)^3} \int_{\mathbb{R}^3}   e^{-2\lambda^2 |\mathbf{p}|^2} e^{i\mathbf{p}\cdot\mathbf{x}} \frac{e^{-i t \sqrt{|\mathbf{p}|^2+m^2}}}{2\sqrt{|\mathbf{p}|^2+m^2}}  e^{- u \sqrt{|\mathbf{p}|^2+m^2}} d\mathbf{p}.
\]
Assuming for the moment $\mathbf{x}=0$, $u=0$ and $t>0$ we get 
\begin{align*}
\Delta_{+,\la}(t,0) &= \frac{e^{-\lambda^2 m^2}}{(2\pi)^3} \int_{\mathbb{R}^3}   e^{-2\lambda^2 |\mathbf{p}|^2} \frac{e^{i t \sqrt{|\mathbf{p}|^2+m^2}}}{2\sqrt{|\mathbf{p}|^2+m^2}} d\mathbf{p}
\\
 &= \frac{e^{-\lambda^2 m^2}}{4\pi^2}  \frac{e^{imt}}{|t|^{\frac{3}{2}}}\int_{0}^\infty  
 e^{-2\lambda^2 \frac{w}{t}(\frac{w}{t}+2m)  }
 {e^{iw}} \sqrt{w}\sqrt{2m+\frac{w}{t}} d w 
\end{align*}
where we have integrated over the angular degrees of freedom of the momentum $\mathbf{p}$, and we have operated a change of coordinates $w=t(\sqrt{|\mathbf{p}|^2+m^2}-m)$.
We have now to estimate an integral of the form
\[
I=\int_{0}^\infty  
 {e^{-iw}} \sqrt{w} g\left(\frac{w}{t}\right) d w 
\]
where $g$ is a rapidly decreasing smooth function. We have that
\[
I=\lim_{\epsilon\to 0^+} g(0) \int_{0}^\infty  
 {e^{-(i+\epsilon)w}} \sqrt{w}  d w 
+ 
 \lim_{\epsilon\to 0^+}\int_{0}^\infty  
 {e^{-(i+\epsilon)w}} \sqrt{w} \left(g\left(\frac{w}{t}\right) - g(0) \right)d w 
\]
the first integrals which contributes to $I$ and the corresponding limit $\epsilon\to 0$ can be computed, we get
\[
I= g(0) \frac{\sqrt{\pi}}{2 i^{\frac{5}{2}}}
+ 
 \lim_{\epsilon\to 0^+}\frac{1}{(i+\epsilon)^2} \int_{0}^\infty  
  \left(e^{-(i+\epsilon)w}-1\right) \frac{d^2}{d w^2 }  \sqrt{w} \left(g\left(\frac{w}{t}\right) - g(0) \right)d w 
\]
where in the second contribution we written
$
{e^{-(i+\epsilon)w}} = \frac{1}{(i+\epsilon)^2}\frac{d^2}{d w^2} \left(e^{-(i+\epsilon)w}-1\right)
$
and we have integrated by parts two times.
The second contribution can be rewritten as follows with a $0<\delta<\frac{1}{2}$. 
\[
\int_{0}^\infty  
\left(e^{-(i+\epsilon)w}-1\right) \frac{d^2}{d w^2 }  \sqrt{w} \left(g\left(\frac{w}{t}\right) - g(0) \right)d w 
=
\int_{0}^\infty  
\left(e^{-(i+\epsilon)w}-1\right) \frac{1}{w^{\frac{3}{2}-\delta}} \frac{1}{t^{\delta}}f_g\left(\frac{w}{t}\right)d w   
\]
where $f_g\left( y\right) = \frac{1}{y^\delta}\left(y^2 g''(y)  + y g'(y) -\frac{1}{4} (g(y)-g(0))\right).$
Notice that $f_g(y)$ is a bounded function, while $\left(e^{-(i+\epsilon)w}-1\right) \frac{1}{w^{\frac{3}{2}-\delta}}$ is bounded by an absolutely integrable function uniformly in $\epsilon$. 
We thus have that 
\begin{equation}\label{eq:estimate-t-beta}
|I| \leq \frac{C_1}{|t|^{\frac{3}{2}}}+\frac{{C}_2}{|t|^{\frac{3}{2}+\delta}} 
\end{equation}
where $C_1$ is a linear function of $g(0)$ and where ${C}_2$ is a linear function of the maximum of $|f|$.  
If an $\mathbf{x} \neq 0$, and $u\in[0,\beta]$ is considered $g(y)$ needs to be multiplied by 
\[
h_{\mathbf{x},u}(y)=\text{sinc}(\sqrt{y(y+2m)}|\mathbf{x}|) e^{-u (y+m)}.
\] 
Notice that $h_{\mathbf{x},u}(y)$ is bounded by a constant uniformly for $u\in[0,\beta]$ and $\mathbf{x}\in \Lambda$ and that the same holds for $f_{gh_{\mathbf{x},u}} $. Therefore $C_1$ and $C_2$ can be chosen independent of $u\in[0,\beta]$ and $\mathbf{x}\in \Lambda$.
If on the other hand $\Delta_{\beta,\lambda}$ is considered we can analyze separately the positive and negative contribution in a similar way. 
In particular we will have to multiply $g$ by suitable functions related to the Bose factors 
\[
b^+_{\mathbf{x},u}(y) = \text{sinc}(\sqrt{y(y+2m)}|\mathbf{x}|) \frac{e^{-u (y+m)}}{1-e^{-\beta (y+m)}},\qquad 
b^-_{\mathbf{x},u}(y) = \text{sinc}(\sqrt{y(y+2m)}|\mathbf{x}|) \frac{e^{u (y+m)}}{e^{\beta (y+m)}-1}
\] 
Again, we have that $gb^{\pm}_{\mathbf{x},u}(0)$ and 
$f_{gb^\pm_{\mathbf{x},u}}$ 
are bounded uniformly for $u\in[0,\beta]$, $\beta \geq \frac 1 m$ and $\mathbf{x}\in \Lambda$.
Hence, the constant $C_1$, ${C}_2$ can be chosen independently on  $u\in[0,\beta]$, $\mathbf{x}\in \Lambda$ and $\beta$ for large $\beta$. This proves the first part of the proposition.
To analyze the second part, we observe that 
\begin{multline*}
\Delta_{\beta,\la}(t,\mathbf{x})-\Delta_{+,\la}(t,\mathbf{x}) = 
\frac{e^{-\lambda^2 m^2}}{(2\pi)^3} \int_{\mathbb{R}^3}   d\mathbf{p}
\frac{ e^{-2\lambda^2 |\mathbf{p}|^2} e^{i\mathbf{p}\mathbf{x}}}{2\sqrt{|\mathbf{p}|^2+m^2}}
\left(e^{-it \sqrt{|\mathbf{p}|^2+m^2}}+ \right.\\
\left. +e^{it \sqrt{|\mathbf{p}|^2+m^2}}    \right)  
\frac{e^{-\beta\sqrt{|\mathbf{p}|^2+m^2}}}{1-e^{-\beta\sqrt{|\mathbf{p}|^2+m^2}}}.
\end{multline*}
We can now repeat the analysis done in the first part of the proof for the positive and negative contribution separately. 
We need furthermore to multiply the function $g$ discussed above with
\[
h(y) = e^{-\beta m }    \frac{e^{-\beta y}}{1-e^{-\beta (y+m)}}
\]
the corresponding $f_{gh}$ is bounded by a constant which multiplies $e^{-\beta m}$ for $\beta m > 1$. This observation allows us to conclude the proof.
\end{proof}

\section{Return to equilibrium on noncommutative spacetime}\label{sec:KMS}

The decay for large values of $t$  of $\Delta_{\beta,\la}(t,\mathbf{x})$ stated in Proposition \ref{pr:clustering} implies 
the following proposition, which is a  
clustering condition which similar to the one stated in Proposition 3.3 in \cite{DFP} for fields on a commutative spacetime.

\begin{proposition}\label{pr:clustering-alfav}
Let $A,B$ in $\mathcal{A}$, consider the interacting time evolution $\alpha_{t}^V$ where $V=V_{\chi,h}$ for some $h$ of compact spatial support and the KMS state introduced in \eqref{eq:free-KMS-state}.
The following clustering condition holds
\begin{equation}\label{eq:clustering-alfav}
\lim_{t\to\infty}\left(\omega_{\la,\beta}(A\star_\la \alpha_t^V(B)) - \omega_{\la,\beta}(A)\omega_{\la,\beta}(\alpha_t^V(B))\right) = 0. 
\end{equation} 
\end{proposition}

\begin{proof}
Consider
\[
D(t)=\left(\omega_{\la,\beta}(A\star_\la \alpha_t^V(B)) - \omega_{\la,\beta}(A)\omega_{\la,\beta}(\alpha_t^V(B))\right),
\]
and notice that differentiating recursively the cocycle condition the interacting time evoultion can be written as
\begin{equation}\label{eq:alfaVB}
\alpha_t^V(B) = \sum_{n=0}^{+\infty}i^n\int_{tS_n}dT\,[K_{t_n},[K_{t_{n-1}},\dots ,[K_{t_1},B_t]]
\end{equation}
where the integration domain is
\begin{equation}\label{eq:tsn}
tS_n := \{(t_1,\dots,t_n) \in \bR^n\,:\, 0<t_n<\dots <t_1<t\}.
\end{equation} 
Hence,  
\begin{equation}\label{eq:sum-DT}
D(t)=
\sum_{n\geq 0}  D_n(t) :=
\sum_{n\geq 0} 
i^n\int_{t S_n} dT \,M \left(  (e^{\Gamma}-1 )\left(\tilde{A} \otimes    [\tilde{K}_{t_n},[\tilde{K}_{t_{n-1}} \dots[\tilde{K}_{t_1},\tilde{B}_t]_{\beta}\dots ]_{\beta}]_{\beta}\right)\right)
\end{equation}
where $M(A\otimes B) := \omega_\la m(A\otimes B)=A(0)B(0)$ (recall that $m(A\otimes B)(\phi)=A(\phi)B(\phi)$), 
\begin{align*}
\Gamma&:= \Gamma_{\Delta_{\beta,\la}} = \int_{M^2}dxdy\, \Delta_{\beta,\la}(x-y) \frac{\delta}{\delta \phi(x) }\otimes \frac{\delta}{\delta \phi(y)},\\
\tilde{A} &:= e^{\frac{1}{2} \int dxdy  (\Delta_{\beta,\la}(x-y) - \Delta_{+,\la}(x-y)) \frac{\delta^2}{\delta{\phi(x)}\delta{\phi(y)}} }A,
\end{align*}
and where   
\begin{equation}\label{eq:commutatior-clustering}
[A,B]_\beta = m(e^{\Gamma} (A \otimes B - B\otimes A)) = m((e^{\Gamma}-1) (A \otimes B - B\otimes A)).
\end{equation}
Hence, the $n-$th element in $D_n(t)$ in the sum in \eqref{eq:sum-DT} can be expanded as a sum over $\mathcal{G}_{n+2}$, the set of connected oriented graphs joining $n+2$ vertices.
The vertices of each graphs in $\mathcal{G}_{n+2}$ are in correspondence with $\{A,K_{t_n}, \dots ,K_{t_1},B\}$ and the edges with $\Gamma_{ij}$ which is $\Gamma$ applied to the $i-$th and $j-$th element of the tensor product $\tilde A \otimes \tilde K_{t_n}\otimes  \dots \otimes \tilde K_{t_1}\otimes \tilde B_t$. 
The only admissible graphs in this graphical expansion are connected because both in \eqref{eq:sum-DT} and in \eqref{eq:commutatior-clustering}  $e^{\Gamma}-1$ appears. 
We have that
\begin{align*}
D_n(t) &= \sum_{G\in{\mathcal{G}_{n+2}}} c_G D_{n,G}(t) \\
&= \sum_{G\in{\mathcal{G}_{n+2}}} c_G\int_{t S_n} dT \,M  \bigg(\prod_{l \in E(G)} \Gamma_{s(l),r(l)}\bigg)\left(\tilde{A} \otimes \tilde{K}_{t_n} \otimes \tilde{K}_{t_{n-1}} \dots \otimes \tilde{K}_{t_1}\otimes \tilde{B}_t\right)
\end{align*}
where $c_G$ is a numerical factor which can be either $i^n$ or $0$, $E(G)$ denotes the set of edges of $G$ and each edge $l\in E(G)$ is a collection of two vertices $l = (s(l),r(l))$ where $s(l),r(l) \in \{0,\dots, n+1\}$. 
Expanding now $\tilde K_{t_j}$ as a sum of terms of the form~\eqref{eq:Kexp} and $A$, $B$ as a sum of terms of the form~\eqref{eq:A}, we see that $D_{n,G}(t)$ can be written as a sum of terms analogous to~\eqref{eq:expansion-differ1}.
Since $A,K,B$ are of compact support, we have that 
$A,K,B\in \cA(\Sigma_{0,a})$ for a sufficiently large $a$.
Now thanks to  
Proposition \ref{pr:time-decay}, and to the fact that $\Delta_{+,\la}$, $\Delta_{F,\la}$ and $\Delta_{\beta,\la}$ are all bounded, we have that, similarly to the proof of Prop.~\ref{prop:limtlimbeta},
\begin{align*}
|D_{n,G}(t)| &\leq E \int_{t S_n} d T \prod_{l \in E(G)}\frac{1}{b+|t_{r(l)}-t_{s(l)}-2a|^{3/2}} 
\leq E' \int_{t S_n} d T \prod_{l \in E(G)}\frac{1}{(|t_{r(l)}-t_{s(l)}|+1)^{3/2}} 
\end{align*}
where $E, E' > 0$ are suitable constants (depending on $G, A, B, a$ but not on $\beta$), $b=(2a+1)^{3/2}$ and $t_0=t$ and $t_{n+1}=0$.
In order to estimate the last integral, let $E_0(G)$ be the set of edges of $G$ connected to the vertex with index 0 (corresponding to $B$), and let $l_0 \in E_0(G)$ be the one among them which is connect to the vertex with minimal index, indicated by $i_0 \in \{1,\dots,n+1\}$. Then, by the form of the integration domain~\eqref{eq:tsn}, $(|t_{r(l_0)}-t_{s(l_0)}|+1)^{-3/2} \leq (t_0-t_1+1)^{-3/2}$ and, for all $l \in E_0(G)$, $(|t_{r(l)}-t_{s(l)}|+1)^{-3/2} \leq (|t_{r(l')}-t_{s(l')}|+1)^{-3/2}$ where $l'$ is an edge which connects $i_0$ with the vertex different from 0 originally attached to $l$. Therefore
\[
\prod_{l \in E(G)}\frac{1}{(|t_{r(l)}-t_{s(l)}|+1)^{3/2}} \leq \frac 1{(t_0-t_1+1)^{3/2}} \prod_{l' \in E(G')}\frac{1}{(|t_{r(l')}-t_{s(l')}|+1)^{3/2}},
\]
where $G'$ is the connected graph obtained by removing the 0 vertex and $l_0$ from $G$ and by replacing all the edges $l \in E_0(G)\setminus\{l_0\}$ with the corresponding $l'$ defined above. Iterating this procedure, one gets in the end
\[
\prod_{l \in E(G)}\frac{1}{(|t_{r(l)}-t_{s(l)}|+1)^{3/2}} \leq \prod_{i=0}^n \frac1 {(t_i-t_{i+1}+1)^{3/2}}.
\]
The integral over $tS_n$ of the last product can now be estimated iteratively thanks to
\[
\int_{0}^{t_{i-1}} \frac{dt_i}{(t_{i-1}-t_i+1)^{\frac{3}{2}}(t_i+1)^{\frac 3 2}} = \frac{4t_{i-1}}{\sqrt{1+t_{i-1}}(2+t_{i-1})^2} 
\leq \frac{4}{(t_{i-1}+1)^{\frac{3}{2}}}, \quad i=1,\dots,n,
\]
which finally implies
\[
|D_{n,G}(t)| \leq \frac{4^n E'}{(t+1)^{\frac{3}{2}}},
\]
and therefore $|D_{n,G}(t)|$ vanishes in the limit of large $t$. 
\end{proof}

The following proposition is a stability result for the interacting equilibrium states.
Similar results are obtained for the ordinary spacetime in Theorem 2 in \cite{BKR} or Theorem 3.3 in \cite{DFP}.

\begin{proposition}\label{pr:return-to-equilibrium}
Consider $V_{\chi,h}$ for some $h$ of compact support. The KMS state $\omega_{\beta,\la}$ of the free theory shows the following return to equilibrium property
\[
\lim_{t\to\infty} \omega_{\beta,\la}(\alpha_t^V(A)) = \omega_{\beta,\la}^{U(i\frac{\beta}{2})}(A), \qquad A\in \mathcal{A},
\]
where $\omega_{\beta,\la}^{U(i\frac{\beta}{2})}$ is a KMS state for the interacting theory defined in \eqref{eq:formal-KMS-state} and in  \eqref{eq:expansion-beta}. 
\end{proposition}

\begin{proof}
We begin by proving that $L(t)=\omega_{\beta,\la}(\alpha_t^V(B))$ is a bounded function of $t$, order by order in perturbation theory. To this end, recalling \eqref{eq:alfaVB} and operating as in \eqref{eq:sum-DT}, we have that
\begin{equation}\label{eq:sum-LT}
L(t)=
\sum_{n\geq 0}  L_n(t) :=
\sum_{n\geq 0} 
i^n\int_{t S_n} dT \,M \left( [\tilde{K}_{t_n},[\tilde{K}_{t_{n-1}} \dots[\tilde{K}_{t_1},\tilde{B}_t]_{\beta}\dots ]_{\beta}]_{\beta}\right)
\end{equation}
where as before also $L_n$ can be expanded as a sum over connected graphs
\begin{align*}
L_n(t) &= \sum_{G\in{\mathcal{G}_{n+2}}} c_G' L_{n,G}(t) \\
&= \sum_{G\in{\mathcal{G}_{n+1}}} c_G'\int_{t S_n} dT \,M  \bigg(\prod_{l \in E(G)} \Gamma_{s(l),r(l)}\bigg)\left(\tilde{K}_{t_n} \otimes \tilde{K}_{t_{n-1}} \dots \otimes \tilde{K}_{t_1}\otimes \tilde{B}_t\right)
\end{align*}
for some suitable constant $c_G'$ which can also be $0$ for some graphs.
Therefore, using iteratively the estimate
\[
\int_{0}^{t_{i-1}}\frac{dt_i}{(t_{i-1}-t_{i}+1)^{3/2}} = 2\left[1-\frac 1 {\sqrt{t_{i-1}+1}}\right] \leq 2
\]
we obtain, similarly to the proof of Prop.~\ref{pr:clustering-alfav},
\begin{equation*}
|L_{n,G}(t)| 
\leq C' \int_{t S_n} d T \prod_{i=1}^{n}\frac{1}{(t_{i-1}-t_{i}+1)^{3/2}}  \leq  2^{n} C' .
\end{equation*}
This proves that every $L_{n}$ is bounded in $t$.

We now follow the strategy of the proof of Theorem 3.3 in \cite{DFP}. 
The KMS conditions implies that
\begin{equation}\label{eq:first-step-return}
\omega_{\la,\beta}(\alpha_t^V(A)) =
\omega_{\la,\beta}(U(t)\star_\la \alpha_t(A)\star_\la U(t)^*) =
\omega_{\la,\beta}(\alpha_t(A)\star_\la U(t)^*\star_\la \alpha_{i\beta}U(t)).
\end{equation}
Now notice that the cocycle condition, the KMS condition and the time translation invariance of the state imply then that 
\begin{align*}
\omega_{\la,\beta}(\alpha_t(A)\star_\la U(t)^*\star_\la \alpha_{s}U(t)) &= 
\omega_{\la,\beta}(\alpha_t(A)\star_\la U(t)^*\star_\la U(s)^*\star_\la U(t)\star_\la \alpha_{t}U(s)   )
\\
&=
\omega_{\la,\beta}(\alpha_{t-i\beta}U(s) \star_\la\alpha_t(A)\star_\la U(t)^*\star_\la U(s)^*\star_\la U(t))
\\
&=
\omega_{\la,\beta}(\alpha_{-i\beta}U(s) \star_\la A\star_\la  \alpha_{-t}^V(U(s)^*)).
\end{align*}
According to the first part of the proof, we can now choose a sequence $\{t_k\}_{k\in \bN} \subset \bR$ converging to $+\infty$ and such that $\lim_{k \to +\infty} \omega_{\la,\beta}(\alpha_{t_k}^V(A))$ exists and is finite. Passing to a subsequence, we can also assume that $N := \lim_{k \to +\infty} \omega_{\la,\beta}( \alpha_{-t_k}^V(U(s)^*))$ is finite. The clustering condition \eqref{eq:clustering-alfav} established in Proposition \ref{pr:clustering-alfav} and the KMS condition imply that 
\begin{equation}\label{eq:second-step-return}\begin{split}
\lim_{k\to\infty} \omega_{\la,\beta}(\alpha_{t_k}(A)\star_\la U(t_k)^*\star_\la \alpha_{s}U(t_k))  &= 
\omega_{\la,\beta}(\alpha_{-i\beta}U(s) \star_\la A) \lim_{k\to\infty}  \omega_{\la,\beta}( \alpha_{-t_k}^V(U(s)^*))
\\
&= 
\omega_{\la,\beta}( A\star_\la U(s) ) N.
\end{split}\end{equation}
The previous equality holds also for $\text{Im}s\in[0,\beta]$: actually, we may extend the results of Proposition \ref{pr:clustering-alfav} to $U(i s)$
following the same proof and using the bounds of Proposition \ref{pr:time-decay} which also hold when some propagator are extended in imaginary time.
The limit $s\to i\beta$ together with \eqref{eq:first-step-return}
gives that
\begin{equation}\label{eq:return}
\lim_{k\to\infty} \omega_{\la,\beta}(\alpha_{t_k}^V(A))   
 =  \omega_{\la,\beta}( A\star_{\la} U(i\beta)) \tilde{N}
 =  \frac{\omega_{\la,\beta}( A\star_{\la} U(i\beta))}{\omega_{\la,\beta}(U(i\beta))}
\end{equation}
where in the last equality we used~\eqref{eq:second-step-return} with $A=1$ and $s=i\beta$, again the KMS condition and the fact that $\omega_{\la,\beta}$ is normalized and hence $\omega_{\la,\beta}(\alpha_t^V(1)) =1$. This shows that the limit on the left hand side of~\eqref{eq:return} is actually independent of the sequence $\{t_k\}$, and therefore
\[
\lim_{t\to\infty} \omega_{\la,\beta}(\alpha_{t}^V(A))   
 =  \frac{\omega_{\la,\beta}( A\star_{\la} U(i\beta))}{\omega_{\la,\beta}(U(i\beta))}.
\] 
The right hand side of the previous equality can be expanded as in \eqref{eq:expansion-beta} with the help of the KMS condition. Actually the KMS condition implies that
\[
\omega_{\la,\beta}( A\star_{\la} U(2s)) = \omega_{\la,\beta}(\alpha_{-i\beta+s}U(s)\star_\la  A\star_{\la} U(s))
\] 
which can be extended to a bounded continuous function for $\text{Im}(s)\in [0,\beta/2]$ analytic in its interior and for $s=\beta/2$  gives the desired result. 
\end{proof}

\textit{Acknowledgements.} G.~M.\ is partially supported by the MIUR Excellence Department Project awarded to the Department of Mathematics, University of Rome Tor Vergata, CUP E83C18000100006, the ERC Advanced Grant 669240 ``Quantum Algebraic Structures and Models'', the INDAM-GNAMPA, and the Tor Vergata University grant ``Operator Algebras and Applications to Noncommutative Structures in Mathematics and Physics''.

\end{document}